\documentclass[12pt]{article}

\usepackage{geometry}
\usepackage[utf8]{inputenc}

\usepackage{amssymb,amsmath,amsthm}
\usepackage{url}
\usepackage{enumerate}

\usepackage[dvipsnames,usenames]{xcolor}

\usepackage{etoolbox,fp}
\usepackage{tikz}

\usetikzlibrary{shapes,
  calc,  
  cd,     
}

\newcommand{\ra}{12.5mm}

\tikzset{
  avertex/.style={regular polygon,regular polygon sides=3,draw,inner sep=1.6pt,fill=white},
  bvertex/.style={regular polygon,regular polygon sides=4,draw,inner sep=2pt},
  vertex/.style={circle,draw,inner sep=2pt,fill=white},
  vvertex/.style={circle,draw,inner sep=2pt,fill=Gray},
}

\theoremstyle{plain}
\newtheorem{theorem}{Theorem}[section]
\newtheorem{lemma}[theorem]{Lemma}

\newtheorem{corollary}[theorem]{Corollary}
\theoremstyle{definition} 
\newtheorem{example}[theorem]{Example}

\newtheorem{definition}[theorem]{Definition}

\newcounter{claim}
\renewcommand{\theclaim}{\Alph{claim}}
\newenvironment{claim}{\refstepcounter{claim}%
\par\medskip\par\noindent{\it Claim~\theclaim.~}~\rm}%
{\par\smallskip\par}
\newenvironment{subproof}{\par\noindent{\sl Proof of Claim~\theclaim.~}}%
{$\,\triangleleft$\par\medskip\par}

\makeatletter
\def\@gifnextchar#1#2#3{\let\@tempe#1\def\@tempa{#2}\def\@tempb{#3}%
  \futurelet\@tempc\@gifnch}

\def\@gifnch{\ifx\@tempc\@sptoken\let\@tempd\@tempb%
  \else\ifx\@tempc\@tempe\let\@tempd\@tempa\else\let\@tempd\@tempb\fi\fi\@tempd}

\def\SK@set#1{\left\{#1\right\}}
\def\SK@@set#1#2{\{#1\,:\,
    \begin{array}{@{}l@{}}#2\end{array}
\}}
\def\SK@mset#1{\left\{\!\!\left\{#1\right\}\!\!\right\}}
\def\SK@@mset#1#2{\{\!\!\{#1\,:\,
    \begin{array}{@{}l@{}}#2\end{array}
\}\!\!\}}
\def\BIG@set#1{\Big\{#1\Big\}}
\def\BIG@@set#1#2{\Big\{#1\:\Big|\:
    \begin{array}{@{}l@{}}#2\end{array}
\Big\}}
\newcommand{\Set}[1]{\@gifnextchar\bgroup{\SK@@set{#1}}{\SK@set{#1}}}
\newcommand{\Mset}[1]{\@gifnextchar\bgroup{\SK@@mset{#1}}{\SK@mset{#1}}}
\newcommand{\Bigset}[1]{\@gifnextchar\bgroup{\BIG@@set{#1}}{\BIG@set{#1}}}
\makeatother

\newcommand{\refeq}[1]{(\ref{eq:#1})}
\newcommand{\of}[1]{\left( #1 \right)}
\newcommand{\function}[2]{:#1 \rightarrow #2}
\newcommand{\bZ}{\mathbb{Z}}
\newcommand{\bQ}{\mathbb{Q}}
\newcommand{\bR}{\mathbb{R}}

\newcommand{\jj}{\mathrm{j}}
\newcommand{\wm}{\mathsf{M}}

\DeclareMathOperator{\rank}{rk}

\newcommand{\barr}{{\bar r}}
\newcommand{\barw}{{\bar w}}
\newcommand{\hatw}{{\hat w}}

\newcommand{\prob}[1]{\mathsf{P}[ #1 ]}

\newcommand{\cprob}[2]{\mathsf{P}[ #1 \mid #2 ]}

\title{Gathering Information about a Graph\\ by Counting Walks from a Single Vertex}

\author{Frank Fuhlbrück\thanks{until Dec. 2023: Institut für Informatik, Humboldt-Universität zu Berlin, Germany; Contact via e-mail: frank.fuhlbrueck@alumni.hu-berlin.de},\quad
  Johannes Köbler\thanks{Institut für Informatik, Humboldt-Universität zu Berlin, Germany.},\quad
Oleg Verbitsky${}^\dagger$\,\thanks{Supported by DFG grant KO 1053/8--2.
  On leave from the IAPMM, Lviv, Ukraine.}\\
and
Maksim Zhukovskii\thanks{Department of Computer Science, University of Sheffield, UK.}
}

\date{}

\begin{document}

\maketitle

\begin{abstract}
  We say that a vertex $v$ in a connected graph $G$ is \emph{decisive} if
  the numbers of walks from $v$ of each length determine the graph $G$ rooted at $v$
  up to isomorphism among all connected rooted graphs with the same number of vertices.
  On the other hand, $v$ is called \emph{ambivalent} if it has the same walk counts
  as a vertex in a non-isomorphic connected graph with the same number of vertices as $G$.
  Using the classical constructions of cospectral trees, we first observe that
  ambivalent vertices exist in almost all trees. If a graph $G$ is determined by spectrum
  and its characteristic polynomial is irreducible, then we prove that all
  vertices of $G$ are decisive. Note that both assumptions are conjectured to be
  true for almost all graphs.
  Without using any assumption, we are able to prove that the vertices of a random graph
  are with high probability distinguishable from each other by the numbers of closed walks
  of length at most 4. As a consequence, the closed walk counts for lengths 2, 3, and 4
  provide a canonical labeling of a random graph.
  Answering a question posed in chemical graph theory,
  we finally show that all walk counts for a vertex in an $n$-vertex graph are determined
  by the counts for the $2n$ shortest lengths, and the bound $2n$ is here asymptotically tight.
\end{abstract}

\section{Introduction}\label{s:intro}

Let $V(G)$ denote the vertex set of a graph $G$.
Given a vertex $v\in V(G)$, we write $G_v$ to denote the rooted version of $G$ where $v$
is designated as a root.
A \emph{vertex invariant} $I$ is a labeling of the vertex set $V(G)$, defined for every graph $G$,
such that the label $I_G(v)$ of a vertex $v\in V(G)$ depends only on the isomorphism type of $G_v$,
that is, $I_G(v)=I_H(\alpha(v))$ for every isomorphism $\alpha$ from $G$ to another graph~$H$.

Given the value $I_G(v)$, how much information can we extract from it about the graph $G$?
In the most favorable case, $I_G(v)$ can yield the isomorphism type of~$G_v$.

\begin{definition}\label{def:I-d-a}
  Let $G$ be a connected graph on $n$ vertices.
  A vertex $v\in V(G)$ is \emph{$I$-decisive} if the equality $I_G(v)=I_H(u)$
  for any other connected $n$-vertex graph $H$ implies that $G_v\cong H_u$.
  On the other hand, a vertex $v\in V(G)$ is \emph{$I$-ambivalent} if there exists
  a connected $n$-vertex graph $H\not\cong G$ with $I_G(v)=I_H(u)$.
\end{definition}

\noindent
Note that $I$-ambivalence is formally a stronger condition than just the negation of $I$-decisiveness.

The questions about the expressibility of vertex invariants comprise problems studied in various areas.

\smallskip

\textit{Isomorphism testing.}
Vertex invariants form the basis of archetypical approaches to the graph isomorphism problem \cite{CorneilK80}
and play an important role in practical implementations \cite{McKayP14}.
The most popular and practical heuristic in the field is \emph{color refinement}.
This algorithm assigns a color $C_G(v)$ to each vertex $v$ of an input graph $G$
and decides that two graphs $G$ and $H$ are non-isomorphic if the multisets
$C(G)=\Mset{C_G(v)}_{v\in V(G)}$ and $C(H)=\Mset{C_H(u)}_{u\in V(H)}$ are different.
If $G$ and $H$ are connected graphs with the same number of vertices, then
the inequality $C(G)\ne C(H)$ actually implies that $C(G)\cap C(H)=\emptyset$.
Consequently, if a connected graph $G$ is identified by color refinement, then
every vertex of $G$ is $C$-decisive.

\smallskip

\textit{Distributed computing.}
A typical setting studied in distributed computing considers a network of processors
that communicate with each other to get certain information about the network topology.
In one communication round, each processor exchanges messages with its neighbors.
In this way, a local information gradually propagates throughout the network.
If the processors do not have identity, Angluin \cite{Angluin80} observed that
this communication process can be well described in terms of color refinement.
In particular, $C_G(v)$ can be understood as all information potentially available
for the processor $v$ in the network $G$. Thus, the decisiveness of $v$ would mean
that this processor is able to completely determine the network topology, provided
the network is connected and its size is known.

\smallskip

\textit{Machine learning.}
Color refinement has turned out to be a useful concept used for comprehending large
graph-structured data \cite{ShervashidzeSLMB11} and for analysis and design of graph
neural networks \cite{MorrisRFHLRG19}. A discussion of vertex invariant based
approaches in this area can be found in~\cite{EliasofT22,MorrisKM17,OrsiniFR15}.

\smallskip

\textit{Local computation.}
Suppose that a random process, like a random walk in a graph, is observed at a single
vertex $v$ of the graph $G$. Which information about the global graph properties can be
recovered from the results of the observation? This question has been investigated
in \cite{BenjaminiTZ23,BenjaminiKLRT06,BenjaminiL02}.
In \cite{BenjaminiKLRT06} it is shown that if the observer records the return time sequence of a random
walk, then the eigenvalues of the graph can be determined under rather general conditions.
Note that the probability distribution studied in \cite{BenjaminiKLRT06}
is determined by the color refinement invariant $C_{G_v}(v)$ where the root $v$ in $G_v$
is individualized by a preassigned special color.

\smallskip

\textit{Mathematical chemistry.} A central concept in the field is the representation
of a chemical compound by \emph{molecular graph} whose vertices correspond to the atoms
and edges to chemical bonds. Chemical compounds are classified based on numerous
invariants of their molecular graphs as, for example, the indices of Estrada, Wiener, Randić
(and many others). A number of vertex invariants are introduced
to serve as atomic descriptors in the molecular graph. One of them, based on the closed walk counts,
was pioneered by Randi\v{c} as a ``diagnostic value for characterization of atomic environment''
\cite{Randic80} and subsequently investigated in the series of
papers \cite{KnopMSRT83,KnopMSTKR86,RandicBNNT89,RandicK89,RandicWG83}
exploiting tight connections of this vertex invariant to spectral graph theory.

\smallskip

In the present paper, we focus on vertex invariants definable in terms of walks.
A \emph{walk} of length $k$
starting at a vertex $v$ in a graph $G$ is a sequence of vertices
$v=v_0,v_1,\ldots,v_k$ such that every two successive vertices $v_i,v_{i+1}$ are adjacent.
If $v_0=v_k$, then the walk is called \emph{closed}.
Let $w_G^k(v)$ denote the total number of walks of length $k$ starting at $v$.
The number of closed walks of length $k$ starting (and ending) at $v$ is denoted by $r^k_G(v)$.
Note that $r^0_G(v)=1$ and $r^1_G(v)=0$.
We define two vertex invariants
\begin{eqnarray*}
  W_G(v)&=&(w_G^1(v),w_G^2(v),\ldots),\\
  R_G(v)&=&(r_G^1(v),r_G^2(v),\ldots)
\end{eqnarray*}
consisting of the counts of walks (resp.\ closed walks) emanating from $v$
for each length $k$. Though $W_G(v)$ and $R_G(v)$ are defined as infinite sequences, they are determined
by a finite number of their first elements; we discuss this issue in the last part of this section.

The aforementioned line of research
\cite{KnopMSRT83,KnopMSTKR86,Randic80,RandicBNNT89,RandicK89,RandicWG83}
in chemical graph theory was motivated, using the terminology of Definition \ref{def:I-d-a}, by the phenomenon of $R$-ambivalence.
Two vertices $v\in V(G)$ and $u\in V(H)$ in molecular graphs $G$ and $H$ are called \emph{isocodal}
if their atomic codes $R_G(v)$ and $R_H(u)$ are equal despite there is no isomorphism
from $G$ to $H$ taking $v$ to $u$. Such vertices were also referred to as \emph{isospectral}
in general and \emph{endospectral} in the particular case of $G=H$.
The terminology is well justified by the fact that the concept of endospectrality is actually
equivalent to the notion of \emph{cospectral vertices} in spectral graph theory
\cite{GodsilMcK76,GodsilM81,GodsilS24,Schwenk73} (see Section \ref{ss:cosp} for details).
The molecular graphs are typically planar, and the case of trees received a special attention
in \cite{KnopMSRT83,KnopMSTKR86,RandicK89}.

It is known that the value of $W_G(v)$ is determined by the color $C_G(v)$ (and,
correspondingly, $R_G(v)$ is determined by $C_{G_v}(v)$).
Different proofs of this fact can be found in \cite{Dvorak10,PowersS82,ESA23}.
As observed in \cite{PowersS82}, the converse does not hold, that is, the vertex invariant
$W_G(v)$ is strictly weaker that $C_G(v)$. Thus, even when a vertex $v$ is known
to be $C$-decisive, we cannot be sure that it is also $W$-decisive.

Demarcating $W$-decisiveness and $W$-ambivalence is one of our main goals.
Since most results will be obtained simultaneously for the two vertex invariants $W$ and $R$,
we use the following simplified terminology.

\begin{definition}\label{def:dec-amb}
  A vertex $v\in V(G)$ is called \emph{decisive} if it is both $W$- and $R$-decisive.
  On the other hand, a vertex $v\in V(G)$ is called \emph{ambivalent} if there exists
  a connected $n$-vertex graph $H\not\cong G$ with both $W_G(v)=W_H(u)$ and $R_G(v)=R_H(u)$.
\end{definition}

We now describe our results, splitting them in four groups.

\smallskip

\textit{\textbf{Ambivalent vertices in trees.}}
The classical result of Schwenk \cite{Schwenk73} says that almost all trees
have cospectral mates. That is, if we take a random labeled tree $T$ on $n$ vertices,
then with probability tending to 1 as $n\to\infty$, there is a tree $S\not\cong T$
having the same eigenvalues, with the same multiplicity, as $T$.
As we already mentioned, there is a tight connection between cospectrality and
closed-walk invariants. Due to this connection, Schwenk's argument immediately
implies that almost all trees contain $R$-ambivalent vertices. We observe that
this extends also to $W$-ambivalence. Using Definition \ref{def:dec-amb},
this result can be stated as follows:
\begin{itemize}
\item Almost every tree has an ambivalent vertex.
\end{itemize}
This is proved in Section \ref{s:trees}, which is to a large extend a survey
of the known relationship between the concept of cospectral vertices and closed-walk count,
Schwenk's proof in \cite{Schwenk73}, and the Harary-Palmer construction of trees
with pseudosimilar vertices \cite{HararyP66}. The last can be seen as the base
of a generic construction of non-isomorphic rooted trees $T_v$ and $S_u$ with
$R_T(v)=R_S(u)$ and $W_T(v)=W_S(u)$. While for the former equality this was known,
for the latter we need some additional analysis carried out in
Lemmas \ref{lem:remsim-walkeq}--\ref{lem:graftage}.

\smallskip

\textit{\textbf{Decisive vertices in general graphs.}}
In Section \ref{s:decisive}, we identify conditions under which all vertices
of a graph are decisive:
\begin{itemize}
\item
  If a graph is determined by spectrum and its characteristic polynomial is irreducible, then
  every vertex of this graph is decisive.
\end{itemize}
Both conditions are fulfilled conjecturally for almost all graphs \cite{HaemersS04,LiuS22,vanDamH03}.
Thus, if these conjectures are true, then the decisiveness of every vertex is
a prevailing graph property. The argument used in Section \ref{s:decisive} is based on the concept
of a \emph{walk matrix} \cite{Godsil12,LiuS22} (see Subsection \ref{ss:wm}), which leads us to a useful
observation that for a vertex $v$ of an $n$-vertex graph $G$, both $W_G(v)$ and $R_G(v)$ are linear
recurrence sequences of order at most $n$. The basics of the theory of linear recurrence,
which we summarize in Subsection \ref{ss:recurr}, turn out to be an efficient tool in the proof.

\smallskip

\textit{\textbf{Local decisiveness within a random graph.}}
The results of Section \ref{s:decisive}, in particular, imply that if the characteristic
polynomial of a graph $G$ is irreducible (which is conjectured to be true for a random graph
with high probability), then $R_G(u)\ne R_G(v)$ for every two vertices $u$ and $v$ of $G$.
In Section \ref{s:canon} we prove this local decisiveness property for a random graph unconditionally.
The similar fact for the vertex invariant $W$ is known. It is an immediate consequence
of the result of O'Rourke and Touri \cite{ORourkeT16}
that the standard walk matrix of a random graph is with high probability non-singular.
As a consequence, both vertex invariants $W$ and $R$ can be used for \emph{canonical labeling}
of a random graph. These facts are, therefore, analogs of the classical result
of Babai, Erd\"{o}s and Selkow \cite{BabaiES80}
saying that the color refinement invariant $C$ produces a canonical labeling for almost all graphs.

The fact that, with high probability, the vertices of a random graph $G$ have pairwise distinct $R$-invariants
has also two consequences for the spectral properties of $G$. First, with high probability $G$
contains no pair of cospectral vertices. The second consequence is interesting in view of the result
obtained by Tao and Vu \cite{TaoVu17} that all eigenvalues $\mu_1,\ldots,\mu_n$ of $G$ are with high probability
pairwise distinct. Let $\alpha_{i,j}$ denote the cosine of the (acute) angle between the $j$-th standard basis vector
and the eigenspace of $\mu_i$. Assume that $\mu_1<\dots<\mu_n$ and that $V(G)=\{1,\ldots,n\}$.
The sequence $\alpha_{1,j},\ldots,\alpha_{n,j}$ is called the \emph{angle sequence} of the vertex $j$
and forms a natural vertex invariant. The absence of cospectral vertices in $G$ can be restated
as the fact that all vertex angle sequences in a random graph are with high probability pairwise distinct
(see Corollaries \ref{cor:no-cospectral} and \ref{cor:no-angles} for details).

In fact, the result proved in Section \ref{s:canon} is much stronger: If $G$ is a random graph on $n$ vertices,
then with probability $1-O(1/\sqrt n)$, every vertex $v$ is distinguished from the other vertices of $G$
by the triple $(r_G^2(v),r_G^3(v),r_G^4(v))$. This is an analog of the result obtained in \cite{ESA23}
for the vertex invariant $W$ saying that, with probability $1-O(\sqrt[4]{\ln n/n})$, every vertex $v$ of $G$
is individualized by the triple $(w_G^1(v),w_G^2(v),w_G^3(v))$.

\smallskip

\textit{\textbf{Bounds for the walking time.}}
Let $v\in V(G)$ and $u\in V(H)$, where $G$ and $H$ are connected graphs on $n$ vertices. To which
value of $k$ should we check the equality $r_G^k(v)=r_H^k(u)$ in order to be sure that
it holds true for all $k$, i.e., $R_G(v)=R_H(u)$? The authors of \cite{RandicWG83}
note that the values $k<n$ are enough in the particular case that $G=H$ and raise the question
of how many first values of $k$ must be checked in general. Curiously, the examples
of $R$-ambivalent vertices found in \cite{RandicWG83} and \cite{KnopMSRT83} were justified
by computing $r_G^k(v)$ and $r_H^k(u)$ for $k\le2n$ (with $n=10$ in \cite{RandicWG83} and $n=16$
in \cite{KnopMSRT83}), with no proof that this upper bound suffices. We cannot find
any continuation of this discussion in the literature and answer the question posed
in \cite{RandicWG83} in Section \ref{s:length}. Our answer retrospectively shows that
the computations made in \cite{RandicWG83} and \cite{KnopMSRT83} are correct.
\begin{itemize}
\item $R_G(v)=R_H(u)$ if and only if $r_G^k(v)=r_H^k(u)$ for all $k<2n$.
  The bound $2n$ is here optimal up to a small additive constant.
\item
  Similarly, $W_G(v)=W_H(u)$ if and only if $w_G^k(v)=w_H^k(u)$ for all $k<2n$.
  The bound $2n$ is here optimal up to an additive term of~$o(n)$.
\end{itemize}
The upper bound of $2n$ is obtained by viewing $R_G(v)$ and $W_G(v)$ as linear
recurrence sequences and by making the general observation that a linear
recurrence sequences of order $n$ is completely determined by the first $2n$ elements.
While the optimality of the bound $2n$ for the vertex invariant $R$ is shown
by a straightforward example, for the vertex invariant $W$ this issue seems to be
more subtle. In this case, we use the graphs that were constructed in \cite{KrebsV15}
in order to answer the similar question for the color refinement invariant~$C$.

\smallskip

The paper is concluded with Section \ref{s:exam} providing several
instructive examples of trees with ambivalent vertices.

\section{Ambivalence in trees}\label{s:trees}

Let $v\in V(G)$ and $u\in V(H)$ be vertices chosen in two graphs $G$ and $H$
(the equality $G=H$ is not excluded).
We call $v$ and $u$ \emph{walk-equivalent} if $W_G(v)=W_H(u)$, that is,
$w_G^k(v)=w_H^k(u)$ for all $k$.
We call vertices $v\in V(G)$ and $u\in V(H)$ \emph{closed-walk-equivalent} if $R_G(v)=R_H(u)$.
We say that $v$ and $u$ are \emph{strongly walk-equivalent} if these vertices
are both walk- and closed-walk-equivalent.

We also recall some well-established terminology.
Two vertices $x$ and $y$ in a graph $G$ are called \emph{similar}
if there is an automorphism $\alpha$ of $G$ such that $\alpha(x)=y$.
Similar vertices are, obviously, walk-equivalent.

We say that \emph{almost every tree} has a property $P$
if the number of labeled trees on $n$ vertices with property $P$
is equal to $(1-o(1))n^{n-2}$, that is, their fraction tends to 1 as $n\to\infty$.

\begin{theorem}\label{thm:trees}
  \hfill

  \begin{enumerate}[1.]
  \item
    Almost every tree has an ambivalent vertex.
  \item
    Almost every tree contains two non-similar strongly walk-equivalent vertices.
  \end{enumerate}
\end{theorem}

Before proving the theorem, we comment on its consequences.
Part 2 shows that the presence of ambivalent vertices is a prevailing phenomenon
not only for pairs of trees but also within a single, randomly chosen tree.
Both parts of the theorem demonstrate an essential difference between the walk-based
vertex invariants $W$ and $R$ from one side and the color refinement invariant $C$
from the other side. As it is well known \cite{ImmermanL90}, every tree is identifiable
by color refinement up to isomorphism and, therefore, all vertices in every tree are $C$-decisive,
in sharp contrast to Theorem~\ref{thm:trees}.

As defined in \cite{LiuS22}, two graphs $G$ and $H$ are \emph{walk-equivalent} if there is a bijection
$\alpha\function{V(G)}{V(H)}$ such that $v$ and $\alpha(v)$ are walk-equivalent for all $v\in V(G)$.
Similarly to the corresponding notion for color refinement, we say that
a graph $G$ is \emph{walk-identifiable} if $G$ is isomorphic to every walk-equivalent $H$.
The proof of Theorem \ref{thm:trees} shows that, in contrast to color refinement,
the identifiability of a graph does not exclude that it contains an ambivalent vertex.

\begin{corollary}\label{cor:walk-id}
  There are walk-identifiable trees with ambivalent vertices.
\end{corollary}

The proof of Theorem~\ref{thm:trees} follows the method developed by Schwenk \cite{Schwenk73}
in his seminal work showing that almost every tree has a non-isomorphic cospectral mate.
Schwenk proved that every fixed rooted graph appears as a \emph{limb} in almost all labeled trees
(see Subsection \ref{ss:proof1} for a formal definition). Another part of Schwenk's argument consists
in finding a limb that ensures the existence of an appropriate mate tree.
The limb used in \cite{Schwenk73} has 9 vertices; see Example~\ref{ex:Schwenk-and-anti}.

In the proof of Theorem \ref{thm:trees}, specifically in Lemmas \ref{lem:remsim-walkeq}--\ref{lem:graftage},
we show that virtually the same approach works also for our purposes. In fact, Schwenk's limb
is quite enough to prove the version of Theorem \ref{thm:trees} restricted to $R$-ambivalence.
The smallest limb ensuring both $R$- and $W$-ambivalence has 11 vertices and is shown in Figure \ref{fig:HP}(a).
This is the smallest tree containing two non-similar walk-equivalent vertices.
The tree was exhibited by Harary and Palmer \cite{HararyP66} as an example of a graph with
two pseudo-similar vertices (see Subsection \ref{ss:remsim}).
It is also used by Godsil and McKay \cite{GodsilMcK76} for strengthening Schwenk's result.
Though the existence of a single tree of this kind, for which the non-similarity and walk-equivalence
of two vertices can be checked by direct computation, is sufficient for proving Theorem \ref{thm:trees},
we take a longer route explaining a general construction of limbs with required properties.

\begin{figure}
  \centering
  \begin{tikzpicture}[every node/.style={circle,draw,inner sep=2pt,fill=none}]
  \begin{scope}[scale=.8]
\path (0,1) node (a1) {}
       (0,0) node (a2) {} edge (a1)
       (1,0) node (a3) {} edge (a2)
       (2,0) node[fill] (a4) {} edge (a3)
       (3,0) node (a5) {} edge (a4)
       (3,1) node (a6) {} edge (a5)
       (4,0) node (a7) {} edge (a5)
       (5,0) node[fill] (a8) {} edge (a7)
       (6,0) node (a9) {} edge (a8)
       (6,1) node (a10) {} edge (a9)
       (7,0) node (a11) {} edge (a9);
       \node[draw=none,fill=none,below] at ($(a4)-(0,0.1)$) {$x$};
       \node[draw=none,fill=none,below] at ($(a8)-(0,0.1)$) {$y$};
       \node[draw=none,fill=none] at (0,-1) {(a)};
   \end{scope}
   \begin{scope}[scale=.8, xshift=.85\textwidth, yshift=5mm]
  \foreach \i in {1,...,9}
  {
    \coordinate (X\i) at (\i*40-10:\ra);
    \node (x\i) at (X\i) {};
    }
    \draw[thick] (x1) -- (x9);
  \foreach \i in {1,2,3,4,5,8}
    {
    \FPadd{\j}{\i}{1}
    \FPround{\j}{\j}{0}
    \draw[thick] (x\i) -- (x\j);
  }
  \draw[thick,gray!50] (x6) -- (x7) -- (x8);
  \node[fill] (y) at (X1) {};
  \node[fill] (x) at (X4) {};
  \node[fill=gray!50] (v) at (X7) {};
  \node[draw=none,fill=none,below] at ($(X7)-(0,0.1)$) {$v$};
  \node[draw=none,fill=none,left] at ($(X4)+(0,0.3)$) {$x=\alpha(v)$};
  \node[draw=none,fill=none,right] at ($(X1)+(0.1,0.2)$) {$y=\alpha^2(v)$};
   \foreach \i in {3,6,9}
   {
     \coordinate (Y\i) at (\i*40-10:1.75*\ra);
     \node (y\i) at (Y\i) {};
    \draw[thick] (x\i) -- (y\i);
  }
  \node[draw=none,fill=none] at (-3,-1.5) {(b)};
   \end{scope}
\end{tikzpicture}
\caption{(a) The Harary-Palmer tree with pseudosimilar (hence non-similar, strongly walk-similar)
  vertices $x$ and $y$.
  (b) The same tree as an instance of the general construction of minimal trees with pseudosimilar
  vertices (obtained by removal of the vertex $v$ from a unicyclic graph with automorphism $\alpha$
  of degree~3).}
  \label{fig:HP}
\end{figure}

\subsection{Closed-walk-equivalent and cospectral vertices}\label{ss:cosp}

For the expository purposes, we briefly explain the connection of the notion
of closed-walk-equivalent vertices
to a closely related notion in spectral graph theory.

Two graphs are \emph{cospectral} if their adjacency matrices have the same spectrum
or, equivalently, the same characteristic polynomial. For a vertex $v$ of a graph $G$,
the \emph{vertex-deleted subgraph} $G\setminus v$ is obtained by
removing $v$ along with all incident edges from $G$.
Two vertices $x$ and $y$ in a graph $G$ are called \emph{cospectral}
if the vertex-deleted subgraphs $G\setminus x$ and $G\setminus y$ are cospectral.
The following fact is well known. It is usually proved by an algebraic argument dating
back to \cite[Lemma 2.1]{GodsilM81}. We here give another,
combinatorial proof based on a characterization of graph cospectrality in terms
of the walk counts. Some ingredients of this argument will be used also later.

\begin{lemma}\label{lem:cosp}
  Two vertices $x$ and $y$ in a graph $G$ are cospectral if and only if
  they are closed-walk-equivalent.
\end{lemma}

\begin{proof}
  Let $\barr^k_G(v)$ denote the number of closed walks of length $k$
  starting at $v$, ending at $v$, and not visiting $v$ meanwhile. Note that
  $\barr^k_G(v)=r^k_G(v)$ for $k\le3$. If $k\ge2$, then
  $$
r^k_G(v)=\sum_{s=2}^k \barr^s_G(v)\,r^{k-s}_G(v).
$$
This easily implies that, for each $k$, the equality $r^s_G(x)=r^s_G(y)$
is true for all $s\le k$ if and only if the equality $\barr^s_G(x)=\barr^s_G(y)$
is true for all $s\le k$.

 Let $R_k(H)=\sum_{v\in V(H)}r_H^k(v)$ denote the total number of closed $k$-walks in
 a graph $H$. It is a well-known folklore result (see, e.g., \cite{GarijoGN11}) that
 graphs $H$ and $K$ are cospectral if and only if $R_k(H)=R_k(K)$ for all $k\ge0$.

 Note that
 \begin{equation}
   \label{eq:GGx}
 R_k(G)=R_k(G\setminus x)+\sum_{s=2}^k s\,\barr^s_G(x)\,r^{k-s}_G(x)
\end{equation}
 for $k\ge2$.
 If $x$ and $y$ are closed-walk-equivalent, then $r^s_G(x)=r^s_G(y)$ and $\barr^s_G(x)=\barr^s_G(y)$
 for all $s$. Along with Eq.~\refeq{GGx} and its version for the vertex $y$, this implies
 that $R_k(G\setminus x)=R_k(G\setminus y)$ for all $k$. Therefore, $G\setminus x$ and $G\setminus y$
 are cospectral. On the other hand, if $x$ and $y$ are cospectral, then
 $R_k(G\setminus x)=R_k(G\setminus y)$ for all $k$ and Eq.~\refeq{GGx}, along with its version for $y$,
 implies that
 $$
\sum_{s=2}^k s\,\barr^s_G(x)\,r^{k-s}_G(x)=\sum_{s=2}^k s\,\barr^s_G(y)\,r^{k-s}_G(y)
$$
for all $k\ge2$.
A simple induction on $k$ shows that, for each $k$, the equalities $\barr^s_G(x)=\barr^s_G(y)$
and $r^s_G(x)=r^s_G(y)$ are true for all $s\le k$. Therefore, $x$ and $y$ are
closed-walk-equivalent.
\end{proof}

\subsection{Removal-similar and pseudosimilar vertices}\label{ss:remsim}

Two vertices $x$ and $y$ in a graph $G$ are called \emph{removal-similar} if
the vertex-deleted subgraphs $G\setminus x$ and $G\setminus y$ are isomorphic.
A survey of the research on this concept is given in \cite{Lauri97}.
Removal-similar vertices are obviously cospectral and, by Lemma \ref{lem:cosp},
closed-walk-equivalent. The following lemma says more.

\begin{lemma}\label{lem:remsim-walkeq}
  Removal-similar vertices are walk-equivalent.
\end{lemma}

\begin{proof}
  Given removal-similar vertices $x$ and $y$ in a graph $G$, we have to prove that
  $w^k_G(x)=w^k_G(y)$ for all~$k$.

  Let $\barw^k_G(v)$ denote the number of walks of length $k$ in $G$
  starting at $v$ and visiting $v$ never again. Note that
  \begin{equation}
    \label{eq:wrnarw}
w^k_G(v)=\sum_{s=0}^k r^s_G(v)\,\barw^{k-s}_G(v).
  \end{equation}
We know, by Lemma \ref{lem:cosp}, that $r^s_G(x)=r^s_G(y)$ for all $s$.
Eq.~\refeq{wrnarw}, therefore, reduces our task to proving that
$$
\barw^k_G(x)=\barw^k_G(y)
$$
for all~$k$.

To this end, let $\hatw^k_G(v)$ denote the number of walks of length $k$ in $G$
visiting the vertex $v$ at least once. Furthermore, let $W_k(G)=\sum_{v\in V(G)}w_G^k(v)$
denote the total number of $k$-walks in $G$. Obviously, $W_k(G)=W_k(G\setminus v)+\hatw^k_G(v)$.
Since the vertices $x$ and $y$ are removal-similar, $W_k(G\setminus x)=W_k(G\setminus y)$.
It follows that
\begin{equation}
  \label{eq:barwxy}
\hatw^k_G(x)=\hatw^k_G(y)
\end{equation}
for all~$k$.

We now prove that $\barw^k_G(x)=\barw^k_G(y)$ by induction on $k$. Note that
\begin{multline*}
\hatw^k_G(v)=\sum_{s=0}^k \sum_{t=0}^{k-s}  \barw^s_G(v)\,r^{k-s-t}_G(v)\,\barw^t_G(v)\\
=2\barw^k_G(v)+
\sum_{s=1}^{k-1} \sum_{t=0}^{k-s}  \barw^s_G(v)\,r^{k-s-t}_G(v)\,\barw^t_G(v) +
\sum_{t=0}^{k-1}  r^{k-t}_G(v)\,\barw^t_G(v).
\end{multline*}
Eq.~\refeq{barwxy}, therefore, implies that
\begin{eqnarray*}
2\barw^k_G(x)&+&
\sum_{s=1}^{k-1} \sum_{t=0}^{k-s}  \barw^s_G(x)\,r^{k-s-t}_G(x)\,\barw^t_G(x) +
\sum_{t=0}^{k-1}  r^{k-t}_G(x)\,\barw^t_G(x)\\
=
2\barw^k_G(y)&+&
\sum_{s=1}^{k-1} \sum_{t=0}^{k-s}  \barw^s_G(y)\,r^{k-s-t}_G(y)\,\barw^t_G(y) +
\sum_{t=0}^{k-1}  r^{k-t}_G(y)\,\barw^t_G(y).
\end{eqnarray*}
It remains to note that the corresponding sums in the left and the right hand sides
of the equality are equal by the induction assumption.
\end{proof}

Similar vertices are obviously removal-similar. Removal-similar but not similar
vertices are called \emph{pseudosimilar}. We call a graph $G$ with a pair of
pseudosimilar vertices \emph{minimal} if no proper subgraph of $G$ contains
such a pair. Harary and Palmer \cite{HararyP66} described a construction
producing trees with removal-similar vertices and proved that every minimal tree with
pseudosimilar vertices can be obtained by this construction.
We recast the Harary-Palmer construction in the style of the more general
construction of graphs with pseudosimilar vertices suggested in \cite{HerndonE75}
and analyzed in~\cite{GodsilK82}.

\medskip

\textsc{The Harary-Palmer construction (recast).}
Let $U$ be a unicyclic connected graph whose automorphism group contains
an element $\alpha$ of order 3. Suppose that a vertex $v$ belongs to the cycle $C$ of $U$ and has degree 2.
Then the vertices $\alpha(v)$ and $\alpha^2(v)$ are removal-similar in~$T=U\setminus v$.
To see this, it is enough to observe that the automorphism $\alpha^2$ of $U$
maps $\{v,\alpha(v)\}$ onto $\{v,\alpha^2(v)\}$. An example is shown in Figure~\ref{fig:HP}(b).

\medskip

Though the construction can sometimes produce a tree with similar vertices,
\cite[Theorem 5]{HararyP66} readily implies that every minimal tree with pseudosimilar vertices
is obtainable in this way. We also remark that a quite constructive description of \emph{all}
trees with pseudosimilar vertices is given in~\cite{KirkpatrickKC83}.

\subsection{Proof of Theorem~\ref{thm:trees}: Part 1}\label{ss:proof1}

Let $G_x$ and $H_z$ be two vertex-disjoint rooted trees.
Their \emph{coalescence} $G_x\cdot H_z$ is a graph obtained from $G$ and $H$
by identifying the root vertices $x$ and $z$. We will keep the name $x$
for the coalesced vertex of $G_x\cdot H_z$.

\begin{lemma}\label{lem:coalesce}
  Let $x$ and $y$ be strongly walk-equivalent vertices in a graph $G$,
  and $z$ be an arbitrary vertex in another graph $H$. Then the vertices
  $x$ in $A=G_x\cdot H_z$ and $y$ in $B=G_y\cdot H_z$ are strongly walk-equivalent.
\end{lemma}

\begin{proof}
  We use the parameters $\barr^k_G(v)$ and $\barw^k_G(v)$ defined in the proofs of Lemma \ref{lem:cosp}
  and Lemma \ref{lem:remsim-walkeq} respectively.

  By assumption, $r^k_G(x)=r^k_G(y)$ for all $k$.
  As noted in the proof of Lemma \ref{lem:cosp}, this implies that
  $\barr^k_G(x)=\barr^k_G(y)$ for all $k$. Let $k\ge2$. Note that
  \begin{eqnarray*}
  r^k_A(x)&=&\sum_{s=2}^k \of{\barr^s_G(x)+\barr^s_H(z)}r^{k-s}_A(x)\text{ and }\\
  r^k_B(y)&=&\sum_{s=2}^k \of{\barr^s_G(y)+\barr^s_H(z)}r^{k-s}_B(y).
  \end{eqnarray*}
  The equality $r^k_A(x)=r^k_B(y)$ follows from here by induction.

  By assumption, we also have $w^k_G(x)=w^k_G(y)$ for all $k$.
  By Eq.~\refeq{wrnarw} in the proof of Lemma \ref{lem:remsim-walkeq}, this implies that
  $\barw^k_G(x)=\barw^k_G(y)$ for all $k$. Note that
  \begin{eqnarray*}
  w^k_A(x)&=&\sum_{s=0}^k \of{\barw^s_G(x)+\barw^s_H(z)}r^{k-s}_A(x)\text{ and }\\
  w^k_B(y)&=&\sum_{s=0}^k \of{\barw^s_G(y)+\barw^s_H(z)}r^{k-s}_B(y).
  \end{eqnarray*}
  The equality $w^k_A(x)=w^k_B(y)$ follows.
\end{proof}

A rooted tree $L_x$ occurs as a \emph{limb} in a tree $T$ if
$T=L_x\cdot M_z$ for some rooted tree $M_z$.
Let $T$ be a random labeled tree on $n$ vertices.
Schwenk \cite{Schwenk73} proved
that every fixed rooted tree $L_x$ occurs as a limb in $T$ with probability $1-o(1)$ as $n\to\infty$.
Fix an arbitrary tree $L$ with non-similar strongly walk-equivalent vertices $x$
and $y$, for example, the Harary-Palmer tree in Figure \ref{fig:HP}(a).
The vertices $x$ and $y$ are pseudosimilar in $L$ and, therefore, they are strongly walk-equivalent
by Lemmas \ref{lem:cosp} and \ref{lem:remsim-walkeq}.
With high probability, $T\cong L_x\cdot M_z$ for this particular $L$ and some rooted tree $M_z$.
Consider $T'=L_y\cdot M_z$. If $n$ is larger than the number of vertices in $L$,
then $T'\not\cong T$ because $T'$ has a smaller number of limbs isomorphic to $L_x$.
Since $x$ and $y$ are strongly walk-equivalent in $L$, their counterparts $x\in V(T)$ and
$y\in V(T')$ are strongly walk-equivalent by Lemma~\ref{lem:coalesce}.

\subsection{Proof of Corollary \ref{cor:walk-id}}\label{ss:compute}

The smallest two trees with ambivalent vertices obtainable by the above construction
have 12 vertices; see Example \ref{ex:W-amenab} below. This proves Corollary \ref{cor:walk-id}
as every tree with at most 24 vertices is walk-identifiable.

The computation certifying the last fact uses the Lua library TCSLibLua in \cite{TCSLibLua} and works as follows.
For each $n$, we trace through all non-isomorphic \emph{rooted} trees on the vertex set $\{1,\ldots,n\}$.
We address only those trees which are rooted at some canonical center;
the other rooted trees are filtered out. These steps are similar to the algorithm outlined in \cite{WrightROM86}.
In this way, we trace through all \emph{unrooted} non-isomorphic trees $T$.
For each tree $T$, we compute a string $s(T)$
encoding the matrix $\wm_T=\of{w_T^k(x)}_{1\le x\le n,\,0\le x<2n}$
(see Section \ref{s:length} for the role of the threshold $2n$).
All trees $T$ for which no collision $s(T)=s(T')$ exists for any other tree $T'$, are walk-identifiable.

This approach can be easily turned into a search for ambivalent vertices
(which yields Example \ref{ex:W-amenab} as the smallest example) by using individual rows of
the matrix $\wm_T$ as keys instead of the entire matrix.

\subsection{Proof of Theorem~\ref{thm:trees}: Part 2}

Given two rooted graphs $G_v$ and $H_u$, define their \emph{graftage} $G_v\curlyvee_w H_u$
as the graph obtained from the disjoint union of $G$ and $H$ by connecting
their vertices $v$ and $u$ to a new vertex $w$.
We can regard the graftage as rooted at $w$ and take its coalescence with another
rooted graph. These operation is a particular case of a more general construction
of graphs with cospectral vertices suggested in \cite{LoweS86} and analyzed in \cite[Proposition~5.1.5]{CvetkovicRS97}.

\begin{lemma}\label{lem:graftage}
  Strongly walk-equivalent vertices $v\in V(G)$ and $u\in V(H)$ remain strongly walk-equivalent
  in the graph $A=(G_v\curlyvee_a H_u)\cdot F_b$ for any rooted graph~$F_b$.
\end{lemma}

\begin{proof}
  By assumption, $r^s_G(v)=r^s_G(u)$ for all $s$. The equalities
  \begin{eqnarray*}
    r^k_A(v)&=&r^k_G(v)+\sum_{s=0}^{k-2}\sum_{t=0}^{k-2-s} r^s_G(v)\,r^{k-2-s-t}_A(a)\,r^t_G(v)\text{ and }\\
    r^k_A(u)&=&r^k_H(u)+\sum_{s=0}^{k-2}\sum_{t=0}^{k-2-s} r^s_H(u)\,r^{k-2-s-t}_A(a)\, r^t_H(u),
  \end{eqnarray*}
  therefore, imply that $r^k_A(v)=r^k_A(u)$ for all $k$. The equalities
  \begin{eqnarray*}
    w^k_A(v)&=&w^k_G(v)+\sum_{s=0}^{k-1} r^s_G(v)\,w^{k-1-s}_A(a)\text{ and }\\
    w^k_A(u)&=&w^k_H(u)+\sum_{s=0}^{k-1} r^s_H(u)\,w^{k-1-s}_A(a)
  \end{eqnarray*}
  now imply that $w^k_A(v)=w^k_A(u)$ for all~$k$.
\end{proof}

To prove Part 2 of Theorem~\ref{thm:trees}, fix an arbitrary tree $L$ with non-similar
strongly walk-equivalent vertices $x$ and $y$, like the Harary-Palmer tree.
Let $L_x$ and $L_y$ be the rooted, vertex-disjoint copies of $L$.
The graftage $L_x\curlyvee_a L_y$ appears as a limb in a random tree $T$ on $n$ vertices
with high probability. The vertices $x$ and $y$ are strongly walk-equivalent in $T$
by Lemma \ref{lem:graftage}. They are non-similar in $T$ because $L_x\not\cong L_y$.
The proof of Theorem~\ref{thm:trees} is complete.

\section{Decisiveness as the average case?}\label{s:decisive}

As it is well known \cite{BabaiES80}, almost every graph is identifiable
by color refinement up to isomorphism and, as a consequence, all vertices in almost every graph are $C$-decisive.
This section is motivated by the question whether the analogous statement holds true for
the vertex invariants $W$ and~$R$.

Since the value of $W_G(v)$ is determined by the value of $C_G(v)$,
every walk-identifiable graph is also identifiable by CR.
The converse is not always true \cite{ESA23}.
Nevertheless, almost all graphs are known to be walk-identifiable \cite{LiuS22,ORourkeT16}.
In view of Corollary \ref{cor:walk-id}, this fact alone does not allow us to conclude
that all vertices of a random graph are $W$-decisive.
This question, also for the vertex invariant $R$, is related to the following two conjectures.

Let $P_G(z)=\det(zI-A)$ denote the characteristic polynomial of a graph $G$.
Here $A$ is the adjacency matrix of $G$ and $I$ is the identity matrix.
A graph $G$ is \emph{determined by spectrum} if $P_G=P_H$ implies $G\cong H$.
Below, irreducibility of a polynomial with integer coefficients is meant over rationals.

\begin{description}
\item[Conjecture A] (see \cite{HaemersS04,vanDamH03})
A random graph is with high probability determined, up to isomorphism, by its spectrum.
\item[Conjecture B] (see \cite[Section 7]{LiuS22})
The characteristic polynomial of a random graph is with high probability irreducible.
\end{description}

\begin{theorem}\label{thm:decisive}
  If $G$ is determined by spectrum and its characteristic polynomial is irreducible, then
  every vertex of $G$ is decisive.
\end{theorem}

\begin{corollary}
    If Conjectures A and B are true, then for almost all $G$, every vertex of $G$ is decisive.
\end{corollary}

To prepare the proof of Theorem \ref{thm:decisive}, we recall some basic stuff
on linear recurrences and present the concept of the walk matrix of a graph.

\subsection{Linear recurrences}\label{ss:recurr}

The following concepts make sense for any field, and we will tacitly consider the rationals.
A \emph{homogeneous linear recurrence relation of order $r$ with constant coefficients
$c_1,\ldots,c_r$} is an equation of the form
\begin{equation}
  \label{eq:lin-recurr}
y_t=c_1y_{t-1}+\cdots+c_ry_{t-r}.
\end{equation}
An infinite sequence $y_0,y_1,\ldots$ satisfying the recurrence relation \refeq{lin-recurr},
is called a \emph{linear recurrence sequence}. The \emph{characteristic polynomial} of
the recurrence relation \refeq{lin-recurr} is defined by $\chi(z)=z^r-c_1z^{r-1}-\cdots-c_{r-1}z-c_r$.

\begin{lemma}[see \cite{EverestPSW03}]\label{lem:recurr}
Let $Y=(y_t)_{t\ge0}$ be a linear recurrence sequence.
Suppose that \refeq{lin-recurr} is a linear recurrence relation of the minimum possible order $r$
satisfied by $Y$, and $\chi$ is the characteristic polynomial of \refeq{lin-recurr}.
For every linear recurrence relation $L'$ with characteristic polynomial $\chi'$,
the following two conditions are equivalent:
\begin{itemize}
\item
  $Y$ satisfies $L'$;
\item
   $\chi$ divides $\chi'$.
\end{itemize}
\end{lemma}

It follows that \refeq{lin-recurr} with minimum possible $r$ is uniquely determined by $Y$,
and we will speak of the \emph{order} $r$
of $Y$ and of the \emph{characteristic polynomial} $\chi_Y$ of~$Y$.

\subsection{Walk matrix}\label{ss:wm}

Without loss of generality we suppose that the vertex set of an $n$-vertex graph
is $\{1,\ldots,n\}$. Given a vertex $x\in V(G)$ and a set of vertices $S\subseteq V(G)$,
let $w_{G,S}^k(x)$ denote the number of walks of length $k$ in $G$ starting at $x$ and
terminating at a vertex in $S$. Following \cite{Godsil12,LiuS22}, we consider
the $n\times n$ matrix $\wm_{G,S}=(m_{x,k})_{1\le x\le n,\,0\le x<n}$
with $m_{x,k}=w_{G,S}^k(x)$ and call it the \emph{walk matrix} of a pair $(G,S)$.
In particular, $\wm_G=\wm_{G,V(G)}$ is the \emph{walk matrix} of a graph~$G$.
If $A$ denotes the adjacency matrix of $G$ and $\jj_S$ is the characteristic vector of $S$,
then the columns of $\wm_{G,S}$ are
$$
\jj_S,A\,\jj_S,A^2\,\jj_S,\ldots,A^{n-1}\,\jj_S.
$$

Let $P_G(z)=z^n-b_1z^{n-1}-\cdots-b_{n-1}z-b_n$ be the characteristic polynomial of $A$.
By the Cayley–Hamilton theorem,
$$
A^n \,\jj_S=b_1A^{n-1} \,\jj_S +\cdots+ b_{r-1}A \,\jj_S + b_n \,\jj_S.
$$
Multiplying both sides of this equality
by $A^{t-n}$ from the left, we see that the vectors of the walk counts
$A^t\,\jj_S=(w_{G,S}^t(1),\ldots,w_{G,S}^t(n))^\top$ satisfy the multidimensional recurrence relation
$$
A^t\,\jj_S=b_1A^{t-1}\,\jj_S+\cdots+b_nA^{t-n}\,\jj_S
$$
of order $n$. Let $x\in V(G)$. As a consequence, the sequence
$$
W_{G,S}(x)=(w_{G,S}^0(x),w_{G,S}^1(x),w_{G,S}^2(x),\ldots)
$$
satisfies the recurrence relation
\begin{equation}
  \label{eq:recurrS}
w_t=b_1w_{t-1}+\cdots+b_nw_{t-n}
\end{equation}
and, therefore, it
is a linear recurrence sequence of order at most $n$.
Denote the characteristic polynomial of this sequence by~$\chi_{G,S,x}$.

\begin{lemma}\label{lem:chiP}
  Let $G$ be a graph with at least two vertices.
  If $P_G$ is irreducible, then $\chi_{G,S,x}=P_G$ for all non-empty $S\subseteq V(G)$ and all $x\in V(G)$.
\end{lemma}

\begin{proof}
  Since the sequence $W_{G,S}(x)$ satisfies the recurrence relation \refeq{recurrS} and
  $P_G$ is the characteristic polynomial of \refeq{recurrS}, we conclude by Lemma \ref{lem:recurr}
  that $\chi_{G,S,x}$ divides $P_G$. The irreducibility of $P_G$ implies that $G$ is connected
  and has at least three vertices. It follows that
  the sequence $W_{G,S}(x)$ is non-constant and, hence, the polynomial $\chi_{G,S,x}$ has degree
  at least 1. The equality $\chi_{G,S,x}=P_G$ now follows from the irreducibility of~$P_G$.
\end{proof}

In general, the recurrence sequence $W_{G,S}(x)$ can have order less than $n$.
Indeed, let $r\le\rank\wm_{G,S}$ be the smallest number such that the vector
$A^r\,\jj_S$ belongs to the span of the vectors $\jj_S,A\,\jj_S,\ldots,A^{r-1}\,\jj_S$ over $\bQ$.
Specifically, let
\begin{equation}
  \label{eq:AAA}
A^r\,\jj_S=a_r\,\jj_S+a_{r-1}A\,\jj_S+\cdots+a_1A^{r-1}\,\jj_S
\end{equation}
for rationals $a_i$. Multiplying both sides of this equality
by $A^{t-r}$ from the left, we obtain the multidimensional recurrence relation
\begin{equation}
  \label{eq:recurr}
A^t\,\jj_S=a_1A^{t-1}\,\jj_S+\cdots+a_rA^{t-r}\,\jj_S
\end{equation}
of order $r$ with characteristic polynomial
\begin{equation}
  \label{eq:main-poly}
M_G(z)=z^r-a_1z^{r-1}-\cdots-a_{r-1}z-a_r.
\end{equation}
It follows by induction that $A^t\,\jj$ belongs to the span of $\jj_S,A\,\jj_S,\ldots,A^{r-1}\,\jj_S$
for all $t\ge r$. From here we conclude that $r=\rank\wm_{G,S}$.\footnote{%
In the case of $S=V(G)$, note that $r$ as well as the sequence $a_1,\ldots,a_r$
are isomorphism-invariant parameters of $G$.
This, in particular, implies that two graphs are walk-equivalent iff their walk matrices
are obtainable from each other by permutation of rows.}
These considerations lead to an alternative proof of the following fact stated in \cite[Corollary~3.8]{LiuS22}.

\begin{lemma}[Liu and Siemons \cite{LiuS22}]\label{lem:rkWM}
  Let $G$ be a graph on $n$ vertices.
  If $P_G$ is irreducible, then $\rank\wm_{G,S}=n$ for all non-empty $S\subseteq V(G)$.
\end{lemma}

\begin{proof}
  As follows from Eq.~\refeq{recurr}, the sequence $W_{G,S}(x)$ satisfies the recurrence relation
  $$
  w_t=a_1w_{t-1}+\cdots+a_rw_{t-r}.
  $$
  Since $M_G$ is the characteristic polynomial of this relation, Lemma \ref{lem:recurr}
  implies that $\chi_{G,S,x}$ divides $M_G$. Since $P_G$ is irreducible, $\chi_{G,S,x}=P_G$
  by Lemma \ref{lem:chiP}. It follows that $M_G$ is divisible by $P_G$ and, hence, the degree
  of $M_G$ is no smaller than the degree of $P_G$, that is, $r\ge n$. We conclude that  $\rank\wm_{G,S}=n$.
\end{proof}

We illustrate the material of this subsection in Example \ref{ex:dist-size} below,
where we also mention its relationship to known concepts of spectral graph theory.

\subsection{Proof of Theorem \ref{thm:decisive}}

Suppose that an $n$-vertex graph $G$ is determined by spectrum and its characteristic polynomial
$P_G$ is irreducible (which implies that $G$ is connected). For a vertex $v\in V(G)$,
assume that $W_G(v)=W_H(u)$ for a vertex $u$ in a connected graph $H$ with $n$ vertices.
Using the analysis in the preceding subsection in the special case of $S=V(G)$, we see that
$W_G(v)$, as well as $W_H(u)$, is a linear recurrent sequence.
Let $\chi_{G,v}=\chi_{G,V(G),v}$ and $\chi_{H,u}=\chi_{H,V(H),u}$ be the characteristic polynomials
of these sequences. Recall that $\chi_{G,v}=\chi_{H,u}$ as the sequences coincide.

We have $\chi_{G,v}=P_G$ by Lemma \ref{lem:chiP}. Lemma \ref{lem:recurr}
implies that $\chi_{H,u}$ divides $P_H$. It follows that $P_G$ divides $P_H$ and, therefore,
$P_G=P_H$. Since $G$ is determined by spectrum, we conclude that $G\cong H$.

By Lemma \ref{lem:rkWM}, $\rank\wm_G=n$.
It follows that the rows of $\wm_G$ are pairwise different. Therefore,
a unique isomorphism from $G$ to $H$ maps $v$ to $u$, which proves that
the vertex $v$ is $W$-decisive.

Before proving the $R$-decisiveness, note that the framework of the preceding subsection
applies also to the vertex invariant $R_G(v)$. Indeed, given a vertex $x\in V(G)$,
let us set $S=\{x\}$. In this case, $w_{G,S}^k(x)=r_G^k(x)$.
The recurrence relation \refeq{recurrS} for $w_t=w_{G,S}^t(x)$, therefore, implies that
the vertex invariant $R_G(x)=(r_G^0(x),r_G^1(x),r_G^2(x),\ldots)$ is a linear recurrence sequence.
Let $\eta_{G,x}=\chi_{G,\{x\},x}$ denote the characteristic polynomial of this sequence.

Assume now that $R_G(v)=R_H(u)$. Using Lemma \ref{lem:chiP}, we conclude like above
that $\eta_{H,u}=\eta_{G,v}=P_G$ divides $P_H$ and, hence, $P_G=P_H$. Since $G$ is determined by spectrum,
we have $G\cong H$.

The graphs $G$ and $H$ can now be identified. The assumption $R_G(v)=R_H(u)$ is therewith converted
in $R_G(v)=R_G(u)$ for a vertex $u\in V(G)$. Set $S=\{u,v\}$. Note $r^k_G(v)=r^k_G(u)$
if and only if $w_{G,S}^k(v)=w_{G,S}^k(u)$. By Lemma \ref{lem:rkWM},
$\rank\wm_{G,S}=n$. It follows that all rows of $\wm_{G,S}$ are pairwise different.
Therefore, the equality $R_G(v)=R_G(u)$ implies that $u=v$.
This proves that the vertex $v$ is $R$-decisive.

\section{Decisiveness within a random graph}\label{s:canon}

The argument at the very end of the proof of Theorem \ref{thm:decisive} shows that
if the characteristic polynomial of a graph $G$ is irreducible, then $G$ does not contain any
pair of closed-walk-equivalent vertices. Under Conjecture B, this is therefore true
for almost all $G$. We here prove this fact unconditionally.

Let $[n]=\{1,\ldots,n\}$. The Erd\H{o}s-R\'enyi random graph $G(n,p)$ is a graph on the vertex set $[n]$
where each pair of distinct vertices $u$ and $v$ is adjacent with probability $p$ independently
of the other pairs. Thus, $G(n,1/2)$ is a random graph chosen uniformly at random from the set of
all graphs on~$[n]$.

\begin{theorem}
  Let $G=G(n,1/2)$. With probability $1-O(n^{-1/2})$,
  every two distinct vertices $u$ and $v$ of $G$ are distinguished by the counts of closed walks
  of length at most 4, that is, $r^k_{G}(u)\neq r^k_{G}(v)$ for at least one $k\in\{2,3,4\}$.
\label{thm:closed-non-equiv-random}
\end{theorem}

Combined with Lemma \ref{lem:cosp}, Theorem \ref{thm:closed-non-equiv-random} immediately implies
the following fact.

\begin{corollary}\label{cor:no-cospectral}
  With probability $1-O(n^{-1/2})$, the random graph $G(n,1/2)$ contains no pair of cospectral vertices.
\end{corollary}

Let $G$ be a graph with $V(G)=[n]$. Denote the adjacency matrix of $G$ by $A$
and let $\mu_1<\mu_2<\ldots<\mu_m$ be all pairwise distinct eigenvalues of $A$,
i.e., all pairwise distinct roots of the characteristic polynomial of $A$.
If all $n$ roots are pairwise distinct (that is, $m=n$), we say that $G$ has \emph{simple spectrum}.
Furthermore, let $E_i$ be the eigenspace of $\mu_i$, i.e., $E_i=\Set{v\in\bR^n}{Av=\mu_iv}$.
For $1\le u\le n$, the standard basis vector $\mathrm{e}_u$ of $\bR^n$ has 1 in the position $u$
and 0 elsewhere. The \emph{angle} $\alpha_{i,u}$ of $G$ is defined to be
the cosine of the angle between $\mathrm{e}_u$ and the eigenspace $E_i$, i.e.,
the angle between $\mathrm{e}_u$ and its projection onto~$E_i$.
Identifying a vertex $u\in V(G)$ with the basis vector $\mathrm{e}_u$, we say that
$\alpha_{1,u},\ldots,\alpha_{m,u}$ is the \emph{angle sequence} of the vertex $u$.
Tao and Vu \cite{TaoVu17} proved that the random graph $G(n,1/2)$ has simple spectrum
with probability at least $1-n^{-c}$, for each fixed $c>0$ and sufficiently large $n$.
Corollary \ref{cor:no-cospectral} implies that, with high probability, not only
all $n$ eigenvalues but also all $n$
vertex angle sequences of $G(n,1/2)$ are pairwise distinct. This readily follows from
the known fact \cite[Proposition 5.1.4]{CvetkovicRS97} that two vertices in a graph are cospectral
if and only if their angle sequences are equal.

\begin{corollary}\label{cor:no-angles}
  With probability $1-O(n^{-1/2})$, no two vertices in $G(n,1/2)$ have the same angle sequences.
\end{corollary}

For the proof of Theorem \ref{thm:closed-non-equiv-random}, we need several probabilistic concentration and anti-concentration bounds.
We say that $X$ is a binomial random variable with parameters $n$ and $p$ and write $X\sim\mathrm{Bin}(n,p)$,
if $X=\sum_{i=1}^n\xi_i$ where $\xi_i$'s are independent Bernoulli random variables, that is,
$X_i=1$ with probability 1 and $X_i=0$ with probability $1-p$.

\begin{lemma}[Chernoff's bound; see, e.g.,~{\cite[Corollary A.1.7]{AlonS16}}]\label{lem:chernoff}
If $X\sim\mathrm{Bin}(n,p)$, then
$$
 \mathbb{P}(|X-np|>t)\le2e^{-2t^2/n}
$$
for every $t\ge 0$.
\end{lemma}

\begin{lemma}[Erd\H{o}s~\cite{Erdos45}, Littlewood and Offord~\cite{LittlewoodO43}]
  Let $X=a_1\xi_1+\ldots+a_n\xi_n$, where $a_1,\ldots,a_n$ are non-zero reals and $\xi_1,\ldots,\xi_n$
  are independent Bernoulli random variables. Then $$\sup_{z\in\mathbb{R}}\prob{X=z}<n^{-1/2}.$$
\label{lem:ELO}
\end{lemma}

\begin{lemma}[Kwan and Sauermann~\cite{KwanS23}]
  Let $X=p(\xi_1,\ldots,\xi_n)$, where $p\in\mathbb{R}[x_1,\ldots,x_n]$ is a polynomial of degree at most 2
  and $\xi_1,\ldots,\xi_n$ be independent Bernoulli random variables.
  If any 0-1-assignment of all but one variables $x_1,\ldots,x_n$ still does not determine the value of
  $p(x_1,\ldots,x_n)$ for the free variable $x_i$ taking value 0 or 1, then $$\sup_{z\in\mathbb{R}}\prob{X=z}=O(n^{-1/2}).$$
\label{lem:ELO2}
\end{lemma}

\noindent
The rest of this section is devoted to the proof of Theorem~\ref{thm:closed-non-equiv-random}.

Let $R^k_G(v)=(r_G^2(v),r_G^3(v),\ldots,r_G^{k+1}(v))$. Let $G=G(n,1/2)$.
By the union bound,
$$
\prob{R_{G}^3(u)=R_{G}^3(v)\text{ for some }u,v} \le
\sum_{u,v}\prob{R_{G}^3(u)=R_{G}^3(v)}
=
{n\choose 2}\prob{R_{G}^3(1)=R_{G}^3(2)}.
$$
Therefore, it suffices to prove that
\begin{equation}
  \label{eq:1=2}
\prob{R_{G}^3(1)=R_{G}^3(2)}=O(n^{-5/2}).
\end{equation}

We denote the set of edges of a graph $H$ by $E(H)$. For $U\subseteq V(H)$, the subgraph of $H$
induced on $U$ is denoted by $H[U]$. For a vertex $v$ in $H$,
we write $\deg^H_U(v)$ for the number of those neighbors of $v$ belonging to $U$.
The neighborhood of $v\in V(H)$ is denoted by~$N_H(v)$.

Let $U_v=N_G(v)$. Note that
\begin{equation}
 r_{G}^2(v)=|U_v|,\quad r_{G}^3(v)=2|E(G[U_v])|,\quad\text{and}\quad
 r_{G}^4(v)=\sum_{w\in[n]}\left(\deg^{G}_{U_v}(w)\right)^2.
\label{eq:r-explicit}
\end{equation}
The first equality in \refeq{r-explicit} follows from the obvious fact that a closed walk of length 2
from/to $v$ is completely determined by an edge incident to $v$. The second equality follows from
the observation that a closed walk of length $3$ is determined by a triangle containg $v$,
which can be walked around in two directions. In order to see the third equality, note that
a closed walk of length 4 consists of a walk of length 2 from $v$ to some vertex $w$ through
a neighbor of $v$ and a return walk from $w$ to $v$ through another, or the same, neighbor of~$v$.

Given sets $U_1,U_2\subseteq V(G)$, define
$$
p(U_1,U_2)=\cprob{r^3_{G}(1)=r^3_{G}(2)\text{ and }r^4_{G}(1)=r^4_{G}(2)}{N_{G}(1)=U_1,\, N_{G}(2)=U_2}.
$$
We have
\begin{equation}
  \prob{R_{G}^3(1)=R_{G}^3(2)}=\sum_{U_1,U_2\,:\,|U_1|=|U_2|}p(U_1,U_2)
  \times
  \prob{N_{G}(1)=U_1,\, N_{G}(2)=U_2}.\label{eq:w3long}
\end{equation}
Note first that
\begin{multline}
  \sum_{|U_1|=|U_2|}\prob{N_{G}(1)=U_1,\ N_{G}(2)=U_2}=\prob{|N_{G}(1)|=|N_{G}(2)|}=\\
  =\prob{|N_{G}(1)\setminus\{2\}|=|N_{G}(2)\setminus\{1\}|}=O(n^{-1/2}).
  \label{eq:first_bound}
\end{multline}
This follows from Lemma \ref{lem:ELO} because $|N_{G}(1)\setminus\{2\}|\sim\mathrm{Bin}(n-2,1/2)$
and $|N_{G}(2)\setminus\{1\}|\sim\mathrm{Bin}(n-2,1/2)$ are independent binomial random variables and,
hence, $|N_{G}(1)\setminus\{2\}|-|N_{G}(2)\setminus\{1\}|$ is a linear combination of $2n-4$ independent
Bernoulli random variables.

Let us call a pair of sets $U_1\subset[n]\setminus\{1\}$ and $U_2\subset[n]\setminus\{2\}$ {\it standard} if
$$
(1/2-n^{-1/4})n\le |U_j|\le (1/2+n^{-1/4})n\text{ for }j=1,2
$$
and
$$
(1/4-n^{-1/4})n\le |U_1\cap U_2|\le (1/4+n^{-1/4})n.
$$
As readily follows from the Chernoff bound (Lemma \ref{lem:chernoff}),
the pair $\of{N_{G}(1),N_{G}(2)}$ is standard with probability $1-e^{-\Omega(\sqrt n)}$.
This implies that the contribution of non-standard pairs $(U_1,U_2)$ in \refeq{w3long}
is negligible, and all what we now have to prove is the estimate
\begin{equation}
  \label{eq:pUU}
p(U_1,U_2)=O(n^{-2})\text{ for all standard }(U_1,U_2)\text{ with }|U_1|=|U_2|,
\end{equation}
where the constant hidden by the big-O notation does not depend on~$(U_1,U_2)$.
Indeed, combining this estimate with \refeq{first_bound} and \refeq{w3long},
we immediately arrive at the desired bound~\refeq{1=2}.

In order to prove \refeq{pUU}, let us fix a standard pair $(U_1,U_2)$ such that $|U_1|=|U_2|$.
In what follows, $H$ will denote a graph on $[n]$ such that
\begin{itemize}
\item $N_H(1)=U_1$, $N_H(2)=U_2$, and
\item there is no edge between $[n]\setminus(U_1\cup U_2\cup\{1,2\})$ and
  $(U_1\setminus (U_2\cup\{2\}))\cup(U_2\setminus (U_1\cup\{1\}))$.
\end{itemize}
Let $G[1,2;U_1,U_2]$ denote the random graph obtained from $G$ by deleting all such edges.
In other words, $G[1,2;U_1,U_2]$ is a version of $G(n,1/2)$ where the edges between
vertices $v\notin U_1\cup U_2\cup\{1,2\}$ and
$u\in (U_1\setminus (U_2\cup\{2\}))\cup(U_2\setminus (U_1\cup\{1\}))$
are not exposed.

For $j=1,2$, let $E_j$ be the set of edges of $H$ induced by $U_j$.
If $E(G[U_j])=E_j$ for $j=1,2$, then  $r_{G}^3(1)=r_{G}^3(2)$ exactly when $|E_1|=|E_2|$. Define
$$
p'(U_1,U_2;H)
=\cprob{r_{G}^4(1)=r_{G}^4(2)}{G[1,2;U_1,U_2]=H}.
$$
Due to~\eqref{eq:r-explicit},
\begin{multline}
 p(U_1,U_2)=\sum_{H:\,|E_1|=|E_2|}p'(U_1,U_2;H)\times\\
 \times\cprob{G[1,2;U_1,U_2]=H}{N_{G}(1)=U_1,\,N_{G}(2)=U_2)},
\label{eq:p_to_p'}
\end{multline}
where the summation goes over all $H$ as specified above satisfying the additional
condition $|E(H[U_1])|=|E(H[U_2])|$.
We first observe that
\begin{multline}
  \sum_{H:\,|E_1|=|E_2|}\cprob{G[1,2;U_1,U_2]=H}{N_{G}(1)=U_1,\,N_{G}(2)=U_2)}\\
  =\cprob{|E(G[U_1])|=|E(G[U_2])|}{N_{G}(1)=U_1,\,N_{G}(2)=U_2)}\\
  =\prob{|E(G[U_1])\setminus E(G[U_1\cap U_2])|=|E(G[U_2])\setminus E(G[U_1\cap U_2])|}=O(n^{-1}).
 \label{eq:second_bound}
\end{multline}
Like the above, this follows from Lemma \ref{lem:ELO} because
$|E(G[U_j])\setminus E(G[U_1\cap U_2])|\sim\mathrm{Bin}\left({|U_j|\choose 2}-{|U_1\cap U_2|\choose 2},1/2\right)$
are independent for $j=1,2$, and the pair $(U_1,U_2)$ is standard.

It remains to prove that
\begin{equation}
  \label{eq:pUUH}
p'(U_1,U_2;H)=O(n^{-1})\text{ for every }H\text{ with }N_H(1)=U_1,\ N_H(2)=U_2,\text{ and } |E_1|=|E_2|,
\end{equation}
where the constant hidden by the big-O notation depends neither on $H$ nor on $U_1$ and~$U_2$.

Given $H$ as specified above, consider a uniformly random graph $G_H$ on $[n]$ such that $G_H[1,2;U_1,U_2]=H$.
 Let
 $$
 \zeta_j=\sum_{w\in[n]}\left(\deg^{G_H}_{U_j}(w)\right)^2\text{ for }j=1,2.
 $$
By~\eqref{eq:r-explicit},
\begin{equation}
p'(U_1,U_2;H)=\prob{\zeta_1=\zeta_2}.
\label{eq:p'_to_poly}
\end{equation}
For $v\notin U_1\cup U_2\cup\{1,2\}$ and $u\in (U_1\setminus(U_2\cup \{2\}))\cup  (U_2\setminus(U_1\cup\{1\}))$,
let $\eta_{v,u}$ denote the indicator random variable of the event that $v$ and $u$ are adjacent in $G_H$.
Since $H$ is fixed, the number $\deg^{G_H}_{U_j}(w)$ is fixed for each $w\in U_1\cup U_2\cup\{1,2\}$.
If $v\notin U_1\cup U_2\cup\{1,2\}$, then
$$
\deg^{G_H}_{U_1}(v)=\sum_{u\in U_1\setminus(U_2\cup\{2\})}\eta_{v,u}+\deg^H_{U_1\cap(U_2\cup\{2\})}(v)=\sum_{u\in U_1\setminus(U_2\cup\{2\})}\eta_{v,u}+\deg^H_{U_1\cap U_2}(v),
$$
and the similar equality holds true for $U_2$. It follows that
$\zeta_1=\zeta_2$ if and only if
\begin{multline*}
  \sum_{v\notin U_1\cup U_2\cup\{1,2\}}\left[
    \left(\sum_{u\in U_1\setminus(U_2\cup\{2\})}\eta_{v,u}\right)^2-
    \left(\sum_{u\in U_2\setminus (U_1\cup\{1\})}\eta_{v,u}\right)^2\,
  \right]\\
  +2\,\sum_{v\notin U_1\cup U_2\cup\{1,2\}}\left[
    \sum_{u\in U_1\setminus(U_2\cup\{2\})}\deg^H_{U_1\cap U_2}(v)\cdot\eta_{v,u}\right.\\
  -\left.
    \sum_{u\in U_2\setminus(U_1\cup\{1\})}\deg^H_{U_1\cap U_2}(v)\cdot\eta_{v,u}
  \right]+N=0
\end{multline*}
for some integer $N=N(H)$. The polynomial in variables $\eta_{v,u}$ on the left-hand side robustly depends
on all variables: If we assign 0-1-values to all variables except $\eta_{v,u}$ for an arbitrary pair $\{v,u\}$,
then the values of the polynomial at $\eta_{v,u}=0$ and $\eta_{v,u}=1$ are different. Lemma \ref{lem:ELO2},
therefore, implies that $\prob{\zeta_1=\zeta_2}=O(n^{-1})$. Using Equality \eqref{eq:p'_to_poly},
we obtain the desired bound \refeq{pUUH}. Along with \eqref{eq:p_to_p'} and \eqref{eq:second_bound},
this yields Bound \refeq{pUU}, completing the proof of the theorem.

\section{How much walking time is necessary?}\label{s:length}

Looking for further analogies between the walk-based vertex invariants $W$/$R$
and the color refinement invariant $C$, recall that the value $C_G(v)$
is iteratively computed as a sequence of vertex colors $C_G^0(v),C_G^1(v),C_G^2(v),\ldots$.
Suppose that $v$ is a vertex in an $n$-vertex graph $G$ and $u$ is a vertex in
an $m$-vertex graph $H$. The standard partition stabilization argument (e.g., \cite{ImmermanL90})
shows that if $C_G^k(v)=C_H^k(u)$ for all $k<n+m$, then this equality holds true for all $k$,
that is, $C_G(v)=C_H(u)$. As shown in \cite{KrebsV15}, the upper bound of $n+m$
is here asymptotically optimal. We now prove the analogs of these facts
for the vertex invariants $W$ and $R$, thereby completing
the discussion initiated in~\cite{RandicWG83} (see Section~\ref{s:intro}).

\begin{theorem}\label{thm:length}
  \hfill

  \begin{enumerate}[\bf 1.]
  \item
    Let $G$ and $H$ be connected graphs on $n$ and $m$ vertices respectively.
    Two vertices $v\in V(G)$ and $u\in V(H)$ are walk-equivalent if and only if
    $w_G^k(v)=w_H^k(u)$ for all $k<n+m$. Similarly, they are closed-walk-equivalent
    if and only if $r_G^k(v)=r_H^k(u)$ for all $k<n+m$.
  \item
    On the other hand, for each $n$ there are $n$-vertex graphs $G$ and $H$
    with vertices $v\in V(G)$ and $u\in V(H)$ such that $v$ and $u$ are
    not closed-walk-equivalent while $r_G^k(v)=r_H^k(u)$  for all $k\le2n-5$.
  \item
    Also, for each $n$ there are $n$-vertex graphs $G$ and $H$
    with vertices $v\in V(G)$ and $u\in V(H)$ such that $v$ and $u$ are
    not walk-equivalent while $w_G^k(v)=w_H^k(u)$ for all $k<2n-16\,\sqrt n$.
\end{enumerate}
\end{theorem}

\subsection{Proof of Theorem \ref{thm:length}: Part 1}

As it was discussed in Section \ref{s:decisive}, if $v$ is a vertex in a graph $G$
with $n$ vertices, then $W_G(v)$ and $R_G(v)$ are linear recurrence sequences of
order at most $n$. Part 1 of the theorem is, therefore, a direct consequence
of the following more general fact.

\begin{lemma}\label{lem:n+m}
  Let $Y=(y_t)_{t\ge0}$ and $Z=(z_t)_{t\ge0}$ be linear recurrence sequences of orders
  $n$ and $m$ respectively. If $y_t=z_t$ for all $t<n+m$, then $Y=Z$.
\end{lemma}

In order to prove Lemma \ref{lem:n+m}, we need a generalization of the concept
of a linear recurrence to higher dimensions. Let $X_0\in\bR^d$ be a vector-column
and $A$ be an $d\times d$ real matrix. The sequence $X_0,X_1,X_2,\ldots$ of $d$-dimensional
vectors satisfying the relation
\begin{equation}
  \label{eq:XtAX}
X_t=AX_{t-1}
\end{equation}
is called a \emph{$d$-dimensional linear recurrence sequence of 1st order}.

\begin{lemma}\label{lem:n-dim-1-ord}
  Let $X_t=(x_{1,t},\ldots,x_{d,t})^\top$ and suppose that $(X_t)_{t\ge0}$ is
  a $d$-dimensional linear recurrence sequence of 1st order satisfying Eq.~\refeq{XtAX}.
  Let $1\le i,j\le d$. If $x_{i,t}=x_{j,t}$ for $t=0,1,\ldots,d-1$, then
  $x_{i,t}=x_{j,t}$ for all~$t\ge0$.
\end{lemma}

\begin{proof}
  Let $P(z)=z^d-a_1z^{d-1}-\cdots-a_{d-1}z-d_n$ be the characteristic polynomial of $A$.
  Similarly to Section \ref{ss:wm}, from the Cayley–Hamilton theorem we derive the equality
$$
A^dX_0=a_1A^{d-1}X_0 +\cdots+ a_{d-1}AX_0 + a_dX_0.
$$
Multiplying both sides of this equality by $A^{t-d}$ from the left and using the induction on $t$,
we conclude that $A^tX_0$, for every $t$, belongs to the linear span of the vectors
$A^{d-1}X_0,\ldots,AX_0,X_0$. Since $X_t=A^tX_0$, this means that $X_t$ is, for every $t$,
a linear combination of $X_{d-1},\ldots,X_1,X_0$. It follows that if the sequences
$(x_{i,t})_{t\ge0}$ and $(x_{j,t})_{t\ge0}$ coincide in the first $d$ positions, then
they coincide everywhere.
\end{proof}

\begin{proof}[Proof of Lemma \ref{lem:n+m}]
  Suppose that the sequences $Y$ and $Z$ satisfy linear recurrence relations
  $y_t=b_1y_{t-1}+\cdots+b_ny_{t-n}$ and $z_t=c_1z_{t-1}+\cdots+c_mz_{t-m}$ respectively.
  Let $B$ be the matrix of the linear transformation of $\bR^n$ mapping a vector
  $(\alpha_0,\alpha_1,\ldots,\alpha_{n-1})$ to the vector
  $(\alpha_1,\alpha_2,\ldots,\alpha_{n-1},b_1\alpha_{n-1}+\cdots+b_n\alpha_0)$. Similarly,
  let $C$ be the matrix of the linear transformation of $\bR^m$ mapping
  $(\alpha_0,\alpha_1,\ldots,\alpha_{m-1})$ to
  $(\alpha_1,\alpha_2,\ldots,\alpha_{m-1},c_1\alpha_{m-1}+\cdots+c_m\alpha_0)$. Note that
  $$
  B(y_0,y_1,\ldots,y_{n-1})^\top=(y_1,y_2,\ldots,y_{n})^\top,\
  B(y_1,y_2,\ldots,y_{n})^\top=(y_2,y_3,\ldots,y_{n+1})^\top,\ldots
  $$
  and, similarly,
  $$
  C(z_0,z_1,\ldots,z_{m-1})^\top=(z_1,z_2,\ldots,z_{m})^\top,\
  C(z_1,z_2,\ldots,z_{m})^\top=(z_2,z_3,\ldots,z_{m+1})^\top,\ldots.
$$
Let $A=B\oplus C$ be the direct sum of the matrices $A$ and $B$.
Set
$$
X_0=(y_0,y_1,\ldots,y_{n-1},z_0,z_1,\ldots,z_{m-1})^\top
$$
and consider the linear recurrence $X_t=AX_{t-1}$.
Let $X_t=(x_{1,t},\ldots,x_{n+m,t})^\top$ and observe that
$x_{1,t}=y_t$ and $x_{n+1,t}=z_t$. Lemma \ref{lem:n+m} now follows by applying
Lemma \ref{lem:n-dim-1-ord} for $i=1$ and $j=n+1$.
\end{proof}

\subsection{Proof of Theorem \ref{thm:length}: Part 2}

Let $G=P_n$ be the path graph on $n$ vertices, and $Y_n$ be the graph
obtained by attaching two new vertices to one of the end vertices of $P_{n-2}$.
Consider $v\in V(G)$ and $u\in V(H)$ as shown in Fig.~\ref{fig:GPnHYn}.

\begin{figure}
  \centering
  \begin{tikzpicture}[every node/.style=vertex,scale=.77]
\path (0,0) node (v1) {}
(0,1) node (v2) {} edge (v1)
(0,2) node (v3) {} edge (v2)
(0,3) node (v4) {} edge (v3)
(0,4) node (v5) {} edge (v4);
\path[xshift=50mm] (0,0) node (u1) {}
(0,1) node (u2) {} edge (u1)
(0,2) node (u3) {} edge (u2)
(-.5,3) node (u4) {} edge (u3)
(0.5,3) node (u5) {} edge (u3);
\node[draw=none,fill=none,below] at (v1) {$v$};
\node[draw=none,fill=none,below] at (u1) {$u$};
\draw[->,dashed] (1.5,2) -- (.2,2);
\node[draw=none,fill=none,right] at (1.7,2) {\it level $\ell$};
\draw[->,dashed] (3.5,2) -- (4.8,2);
\end{tikzpicture}
  \caption{The graphs $G=P_n$ and $H=Y_n$ for $n=5$.}
  \label{fig:GPnHYn}
\end{figure}

Set $\ell=n-3$ and let $G_\ell$ and $H_\ell$ denote
the subgraphs of $G$ and $H$ spanned by the sets of vertices at the distance at most
$\ell$ from $v$ and $u$. Note that $G_\ell\cong H_\ell\cong P_{n-2}$.
We have $r_G^k(v)=r_H^k(u)$ for all $k\le2\ell+1=2n-5$ because all closed walks
from $v$ and $u$ of length at most $2\ell+1$ are contained in $G_\ell$ and $H_\ell$ respectively.
Nevertheless, $v$ and $u$ are not closed-walk-equivalent because $r_G^{2n-4}(v)\ne r_H^{2n-4}(u)$.
The inequality follows from equalities
\begin{eqnarray}
  r_G^{2\ell+2}(v)&=&r_{G_\ell}^{2\ell+2}(v)+1\text{ and}\label{eq:rG2ellv}\\
  r_H^{2\ell+2}(u)&=&r_{H_\ell}^{2\ell+2}(u)+2\label{eq:rH2ellu}
\end{eqnarray}
as $r_{G_\ell}^{2\ell+2}(v)=r_{H_\ell}^{2\ell+2}(u)$.

\subsection{Proof of Theorem \ref{thm:length}: Part 3}

The proof of this part is not as easy as it was for closed walks.
Note that we cannot use the same graphs as, similarly to \refeq{rG2ellv}--\refeq{rH2ellu},
\begin{eqnarray*}
  w_G^{n-2}(v)&=&w_{G_\ell}^{n-2}(v)+1\text{ and}\\
  w_H^{n-2}(u)&=&w_{H_\ell}^{n-2}(u)+2
\end{eqnarray*}
and, therefore, $w_G^{n-2}(v)\ne w_H^{n-2}(u)$. The technical challenge is to ensure that the
top levels of $G$ and $H$ are still, at least locally, indistinguishable by walk counts.

\begin{figure}
  \centering
\begin{tikzpicture}[scale=.77]
\newcommand{\diamoond}[3]{
\path (#1,#2) node[bvertex] (bb#3) {}
    (#1,#2+2) node[bvertex] (bt#3) {} edge (bb#3);
}
\newcommand{\diamoondrev}[3]{
\path (#1,#2) node[bvertex] (bb#3) {}
    (#1,#2+2) node[bvertex] (bt#3) {} edge (bb#3);
}
\newcommand{\tail}[3]{
\path (#1,#2) node[vertex] (vb#3) {}
      (#1,#2+1.25) node[vvertex] (vvb#3) {} edge (vb#3)
      (#1,#2+2.5) node[vvertex] (vvt#3) {} edge (vvb#3)
      (#1,#2+3.75) node[vertex] (vt#3) {} edge (vvt#3);
}
\newcommand{\bloock}{
\tail001
\diamoond{-1.5}61
\diamoondrev{1.5}62
\path (0,5) node[avertex] (ab) {} edge (vt1) edge (bb1) edge (bb2);
}

\begin{scope}[scale=.5]
\bloock
\path (0,9) node[avertex] (at) {} edge (bt1) edge (bt2)
      (0,9.5) node[inner sep=0pt,fill=black] () {} edge (at)
      (0,-.5) node[inner sep=0pt,fill=black] () {} edge (vb1);
\node[draw=none,fill=none] at (-2,0) {(a)};
\end{scope}

\begin{scope}[scale=.5,xshift=90mm]
\bloock
\tail0{10.125}2
\path (0,9) node[avertex] (at) {} edge (bt1) edge (bt2) edge (vb2)
      (0,15) node[avertex] (atop) {} edge (vt2)
      (-1.5,16.5) node[bvertex] (bbl) {} edge (atop)
      (1.5,16.5) node[bvertex] (bbr) {} edge (atop) edge (bbl)
(0,-.5) node[inner sep=0pt,fill=black] () {} edge (vb1);
\node[draw=none,fill=none] at (-2,0) {(b)};
\end{scope}

\begin{scope}[scale=.5,xshift=190mm]
\bloock
\path (-3,10) node[avertex] (al) {} edge (bt1)
      (-1.8,10) node[vertex] (vl) {} edge (al)
      (-.6,10) node[vvertex] (vvl) {} edge (vl)
      (.6,10) node[vvertex] (vvr) {} edge (vvl)
      (1.8,10) node[vertex] (vr) {} edge (vvr)
      (3,10) node[avertex] (ar) {} edge (bt2)  edge (vr)
     (-1.5,13) node[bvertex] (bbl) {} edge (al)
      (1.5,13) node[bvertex] (bbr) {} edge (ar) edge (bbl)
(0,-.5) node[inner sep=0pt,fill=black] () {} edge (vb1);
\node[draw=none,fill=none] at (-2,0) {(c)};
\end{scope}

\end{tikzpicture}
  \caption{(a) The tail block $T_{6,s}$ for $s=4$.
(b) The head block of $G_{4,t}$.
(c) The head block of~$H_{4,t}$.}
  \label{fig:port}
\end{figure}

We make use of the construction of a pair of graphs $G_{s,t}$ and $H_{s,t}$ from \cite{KrebsV15},
where $s\ge1$ and $t\ge2$ are integer parameters.
Each of the graphs is a chain of $t$ blocks.
The first $t-1$ \emph{tail blocks} are all the same in both $G_{s,t}$ and $H_{s,t}$.
Each tail block is a copy of the tadpole graph $T_{6,s}$, i.e.,
is obtained from the cycle $C_6$ and the path $P_s$ by adding an edge between
an end vertex of the path and a vertex of the cycle; see Fig.~\ref{fig:port}(a).
In addition, each of the graphs contains one \emph{head block},
and the head blocks of $G_{s,t}$ and $H_{s,t}$ are different; see Fig.~\ref{fig:port}(b--c).
Note that the head block of $G_{s,t}$ contains a subgraph isomorphic to~$T_{6,s}$.

\begin{figure}
  \centering
\begin{tikzpicture}[scale=.77]
\newcommand{\diamoond}[3]{
\path (#1,#2) node[bvertex] (b#3) {}
    (#1,#2+1.75) node[bvertex] (bb#3) {} edge (b#3);
}
\newcommand{\bloock}[3]{
\diamoond{#1-1.5}{#2+3}{1#3}
\diamoond{#1+1.5}{#2+3}{2#3}
\path (#1,#2-.25) node[vertex] (v#3) {}
    (#1,#2+.5) node[vvertex] (vnew) {} edge (v#3)
    (#1,#2+1.25) node[vertex] (vv) {} edge (vnew)
    (#1,#2+2) node[avertex] (a) {} edge (b1#3) edge (b2#3) edge (vv)
    (#1,#2+5.75) node[avertex] (a#3) {} edge (bb1#3) edge (bb2#3);
   }

\begin{scope}[scale=.5]
\bloock001
\bloock072
\bloock0{14}3
\path (0,20.75) node[vertex] (v4) {}
      (0,21.5) node[vvertex] (vnew45) {} edge (v4)
      (0,22.25) node[vertex] (v5) {} edge (vnew45)
      (0,23) node[avertex] (a4) {} edge (v5)
   (-1.5,24.5) node[bvertex] (b4) {} edge (a4)
    (1.5,24.5) node[bvertex] (bb4) {} edge (a4) edge (b4);
 \draw (v2) -- (a1) (v3) -- (a2) (v4) -- (a3);
\node[draw=none,fill=none,below] at (0,-.55) {$v$};
\node[draw=none,fill=none,left] at (bb13) {$a$};
\node[draw=none,fill=none,right] at (bb23) {$b$};
\node[draw=none,fill=none,below] at (a3) {$c$};
\node[draw=none,fill=none,left] at (v4) {$g$};
\draw[->,dashed] (5,18.75) -- (2.65,18.75);
\node[draw=none,fill=none,right] at (5.75,18.75) {\it level $\ell$};
\end{scope}

\begin{scope}[scale=.5,xshift=150mm]
\bloock001
\bloock072
\diamoond{-1.5}{17}x
\diamoond{1.5}{17}y
\path (0,13.75) node[vertex] (v3) {} edge (a2)
      (0,14.5) node[vvertex] (vnew3) {} edge (v3)
      (0,15.25) node[vertex] (vv3) {} edge (vnew3)
      (0,16) node[avertex] (aaa) {} edge (vv3) edge (bx) edge (by)
   (-3,21) node[avertex] (aaa1) {} edge (bbx)
    (3,21) node[avertex] (aaa2) {} edge (bby)
    (-1.5,21) node[vertex] (vvv1) {} edge (aaa1)
     (0,21) node[vvertex] (vvvnew) {} edge (vvv1)
     (1.5,21) node[vertex] (vvv2) {} edge (aaa2) edge (vvvnew)
    (-1.5,24) node[bvertex] (b4) {} edge (aaa1)
     (1.5,24) node[bvertex] (bb4) {} edge (aaa2) edge (b4);
  \draw (v2) -- (a1);
\draw[dashed] (.5,13.5) -- (8,13.5);
\node[draw=none,fill=none,above] at (6,13.5) {\it the head block};
\node[draw=none,fill=none,below] at (6,13.5) {\it $t-1$ tail blocks};
\node[draw=none,fill=none,below] at (0,-.55) {$u$};
\node[draw=none,fill=none,left] at (bbx) {$a'$};
\node[draw=none,fill=none,right] at (bby) {$b'$};
\node[draw=none,fill=none,left] at (aaa1) {$c'$};
\node[draw=none,fill=none,right] at (aaa2) {$c''$};
\node[draw=none,fill=none,above] at (vvv1) {$g'$};
\node[draw=none,fill=none,left] at (b4) {$b''$};
\node[draw=none,fill=none,above] at (vvvnew) {$y'$};
\draw[->,dashed] (-5,18.75) -- (-3,18.75);
\end{scope}

\end{tikzpicture}
  \caption{The graphs $G_{s,t}$ and $H_{s,t}$ for $s=t=3$.}
  \label{fig:GstHst}
\end{figure}

An example of the construction for $s=t=3$ is shown in Fig.~\ref{fig:GstHst}.
Note that both $G_{s,t}$ and $H_{s,t}$ have $n=t(s+6)+s+3$ vertices.

The first tail block in each of the graphs $G=G_{s,t}$ and $H=H_{s,t}$ contains a vertex of degree 1,
which is a single vertex of degree 1 in the graph. Those are
the vertices $v\in V(G)$ and $v\in V(H)$, for which we claim that
\begin{enumerate}[(i)]
\item
  $w_G^k(v)=w_H^k(u)$ for all $k<2t(s+4)-1$ and
\item
  $w_G^k(v)\ne w_H^k(u)$ for $k=2t(s+4)-1$.
\end{enumerate}
Part 3 of the theorem follows by setting $s=3t$. More precisely, Conditions (i)--(ii)
imply Part 3 for $n$ of the form $n=3t^2+9t+3$. For all other $n$, we use $G_{3t,t}$ and $H_{3t,t}$
with the largest $t$ such that $3t^2+9t+3<n$ and attach the missing number of new degree-1 vertices $n-(3t^2+9t+3)$
to $v$ and to~$u$.

Let us proceed to proving Claims (i)--(ii).
The graphs $G$ and $H$ are uncolored. However, Figures \ref{fig:port} and \ref{fig:GstHst}
show an auxiliary coloring of the vertices, which has an important feature.

\smallskip

\textit{Property (*)}:
With the exception of $u$ and $v$, any two vertices of the same color have
the same number of neighbors of each color. For example, if $s=3$, then there are altogether
four possible neighborhood patterns:
$$
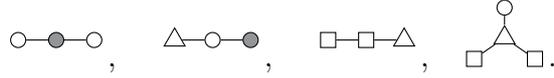
\raisebox{2.5mm}{
\begin{tikzpicture}[scale=.5]
\path (0,0) node[vertex] (a) {}
      (1,0) node[vvertex] (b)  {} edge (a)
      (2,0) node[vertex] (c)  {} edge (b);
\end{tikzpicture}}\,,\quad
\raisebox{2.5mm}{
\begin{tikzpicture}[scale=.5]
\path (0,0) node[avertex] (a) {}
      (1,0) node[vertex] (b)  {} edge (a)
      (2,0) node[vvertex] (c)  {} edge (b);
    \end{tikzpicture}}\,,\quad
\raisebox{2.5mm}{
  \begin{tikzpicture}[scale=.5]
\path (0,0) node[bvertex] (a) {}
      (1,0) node[bvertex] (b)  {} edge (a)
      (2,0) node[avertex] (c)  {} edge (b);
    \end{tikzpicture}}\,,\quad
\begin{tikzpicture}[scale=.4]
\path (-1,-.7) node[bvertex] (a) {}
      (0,0) node[avertex] (b)  {} edge (a)
      (1,-.7) node[bvertex] (c)  {} edge (b)
      (0,1) node[vertex] (d)  {} edge (b);
\end{tikzpicture}\,.
$$

Let $d(a,b)$ denote the distance between vertices $a$ and $b$ in a graph.
We will also write $N_G(v)$ to denote the neighborhood of a vertex $v$ in a graph~$G$.

\begin{claim}\label{cl:1}
  Assume that vertices $x\in V(G)$ and $y\in V(H)$ are equally colored. Then
  $w_G^k(x)=w_H^k(y)$ for all $k\le\min\of{d(x,v),d(y,u)}$.
\end{claim}

\begin{subproof}
  We prove, by finite induction on $k$, that $w_G^k(x)=w_H^k(y)$ for all $x\in V(G)$ and $y\in V(H)$
  such that both $d(x,v)$ and $d(y,u)$ are no less than $k$. The equality is trivially true for $k=0$.
  If $k\ge1$, then $w_G^k(x)=\sum_{z\in N_G(x)}w_G^{k-1}(z)$ and, similarly,
  $w_H^k(y)=\sum_{z\in N_H(y)}w_H^{k-1}(z)$. These sums are equal because, according to Property (*),
  each color appears in the neighborhoods $N_G(x)$ and $N_H(y)$ with the same multiplicity and
  the neighbors of $x$ and $y$ can be only 1 closer to the vertices $v$ and $u$ respectively.
\end{subproof}

\begin{claim}\label{cl:2}
  Let $x\in V(G)$ and $y\in V(H)$ be equally colored vertices.
Assume that $d(x,v)\ne d(y,u)$ and set $h=\min\of{d(x,v),d(y,u)}$.
Then $w_G^{h+1}(x)\ne w_H^{h+1}(y)$.
\end{claim}

\begin{subproof}
  We proceed by induction on $h$. If $h=0$, then either $x=v$ and $y\ne u$ or
  $x\ne v$ and $y=u$. In the latter case, $w_H^{1}(y)=1$ while $w_G^{1}(x)=2$.
  The former case is similar.

  If $h\ge1$, then $w_G^{h+1}(x)=\sum_{z\in N_G(x)}w_G^{h}(z)$ and, similarly,
  $w_H^{h+1}(y)=\sum_{z\in N_H(y)}w_H^{h}(z)$. To be specific, suppose that
  $d(x,v)>d(y,u)=h$; the analysis of the other case is similar.
  Fix a color-preserving bijection $f\function{N_H(y)}{N_G(x)}$.
  The neighborhood $N_H(y)$ contains a single vertex $e$ closer to $u$ than $y$, i.e.,
  such that $d(e,u)=h-1$. In fact, the vertex $y'$ shown in Fig.~\ref{fig:GstHst}
  has two such neighbors, but the inequality $d(x,v)>d(y',u)$ is then impossible
  for any vertex $x\in V(G)$ of the same color.
  By the induction assumption, $w_G^{h}(f(e))\ne w_H^{h}(e)$.
  For any other neighbor $z\ne s$ of $y$, we have $w_G^{h}(f(z))=w_H^{h}(z)$
  by Claim \ref{cl:1}. The inequality $w_G^{h+1}(x)\ne w_H^{h+1}(y)$ follows.
\end{subproof}

Let $\ell=t(s+4)-2$ be the level up to which the graphs $G$ and $H$ are isomorphic;
see Fig.~\ref{fig:GstHst}. More formally, let $G_\ell$ and $H_\ell$ denote
the subgraphs of $G$ and $H$ spanned by the sets of vertices at the distance at most
$\ell$ from $v$ and $u$. Then $\ell$ is defined as the maximum integer such that
$G_\ell$ and $H_\ell$ are isomorphic.

Let $a$ and $b$ be the vertices at level $\ell$ in $G$ and $a'$ and
$b'$ be the vertices at level $\ell$ in $H$. Furthermore, let $c$ be
the vertex at the next level in $G$ and $c'$ and $c''$ be the vertices
at the next level in $H$; see Fig.~\ref{fig:GstHst}.

We are now prepared to prove Claim~(i).

\begin{claim}\label{cl:3}
  $w_G^k(v)=w_H^k(u)$ for all $k\le2\ell+2$.
\end{claim}

\begin{subproof}
  Let $w_G^k(x,z)$ denote the number of walks of length $k$ in $G$ starting at a vertex $x$
  and ending at a vertex~$z$.

  A walk of length $k$ from $v$ either is entirely contained in $G_\ell$ or leaves $G_\ell$
  after $m\ge\ell$ steps through $a$ or $b$, arrives at $c$, and then follows
  some walk of length $k-m-1$ from $c$ in $G$. Using this and the similar observation for $H$,
  we obtain the equalities
  \begin{eqnarray}
    w_G^k(v)&=&w_{G_\ell}^k(v)+\sum_{m=\ell}^{k-1}
                \of{ w_{G_\ell}^m(v,a) w_G^{k-m-1}(c) + w_{G_\ell}^m(v,b) w_G^{k-m-1}(c) },\label{eq:wGkv}\\
    w_H^k(u)&=&w_{H_\ell}^k(u)+\sum_{m=\ell}^{k-1}
                \of{ w_{H_\ell}^m(u,a') w_H^{k-m-1}(c') + w_{H_\ell}^m(u,b') w_H^{k-m-1}(c'') }.\label{eq:wHku}
  \end{eqnarray}
  Since $G_\ell$ and $H_\ell$ are isomorphic, we have
  $$
  w_{G_\ell}^k(v)=w_{H_\ell}^k(u),\quad
  w_{G_\ell}^m(v,a)=w_{H_\ell}^m(u,a'),\text{ and }w_{G_\ell}^m(v,b)=w_{H_\ell}^m(u,b').
  $$
  Since $m\ge\ell$ and $k\le2\ell+2$, we have
  $$
  k-m-1\le\ell+1=d(c,v)=d(c',u)=d(c'',u)
  $$
  and, therefore,
  $$
w_G^{k-m-1}(c)=w_H^{k-m-1}(c')=w_H^{k-m-1}(c'')
$$
by Claim \ref{cl:1}. We conclude that $w_G^k(v)=w_H^k(u)$.
\end{subproof}

\begin{claim}\label{cl:4}
 $w_G^{\ell+2}(c)\ne w_H^{\ell+2}(c')=w_H^{\ell+2}(c'')$.
\end{claim}

\begin{subproof}
We begin with equalities
  \begin{eqnarray*}
    w_G^{\ell+2}(c)&=& w_G^{\ell+1}(a)+ w_G^{\ell+1}(b)+ w_G^{\ell+1}(g),\\
    w_H^{\ell+2}(c')&=& w_H^{\ell+1}(a')+ w_H^{\ell+1}(b'')+ w_H^{\ell+1}(g');
  \end{eqnarray*}
  see Fig.~\ref{fig:GstHst}.
  Note first that $w_G^{\ell+1}(a)=w_H^{\ell+1}(a')$. Indeed,
   \begin{eqnarray*}
     w_G^{\ell+1}(a)&=&w_{G_\ell}^{\ell+1}(a)+\sum_{m=0}^{\ell}
                     \of{ w_{G_\ell}^m(a,a) w_G^{\ell-m}(c) + w_{G_\ell}^m(a,b) w_G^{\ell-m}(c) },\\
     w_H^{\ell+1}(a')&=&w_{H_\ell}^{\ell+1}(a')+\sum_{m=0}^{\ell}
                     \of{ w_{H_\ell}^m(a',a') w_H^{\ell-m}(c') + w_{H_\ell}^m(a',b') w_H^{\ell-m}(c'') },
  \end{eqnarray*}
  where $w_G^{\ell-m}(c)=w_H^{\ell-m}(c')=w_H^{\ell-m}(c'')$ by Claim \ref{cl:1}
  and the other corresponding terms are equal due to the isomorphism $G_\ell\cong H_\ell$.
  Furthermore, $w_G^{\ell+1}(g)=w_H^{\ell+1}(g')$ by Claim \ref{cl:1},
  and $w_G^{\ell+1}(b)\ne w_H^{\ell+1}(b'')$ by Claim \ref{cl:2}.
  We conclude that $w_G^{\ell+2}(c)\ne w_H^{\ell+2}(c')$.
\end{subproof}

We now prove Claim~(ii).

\begin{claim}\label{cl:5}
 $w_G^{2\ell+3}(v)\ne w_H^{2\ell+3}(u)$.
\end{claim}

\begin{subproof}
  The inequality follows from Eqs.~\refeq{wGkv}--\refeq{wHku} for $k=2\ell+3$ after noting that
  $$
w_G^{l}(c)=w_H^{l}(c')=w_H^{l}(c'')
$$
for all $l\le\ell+1$ by Claim \ref{cl:1} and that
$$
w_G^{\ell+2}(c)\ne w_H^{\ell+2}(c')=w_H^{\ell+2}(c'')
$$
by Claim \ref{cl:4}.
\end{subproof}

The proof of Theorem \ref{thm:length} is complete.

\section{Ambivalence in examples}\label{s:exam}

We here give examples of particular graphs illustrating the material of the previous sections.
The computations are performed with the use of the library TCSLibLua \cite{TCSLibLua}; cf.\ Subsection~\ref{ss:compute}.

\begin{example}[The smallest trees with ambivalent vertices]\label{ex:W-amenab}
  The non-isomorphic trees $T$ and $S$ depicted below are obtained as described
  in the proof of Part 1 of Theorem~\ref{thm:trees}.

\bigskip

\noindent
\begin{tikzpicture}[every node/.style={circle,draw,inner sep=2pt,fill=none}]
  \begin{scope}[scale=.8]
\path (0,1) node (a1) {}
       (0,0) node (a2) {} edge (a1)
       (1,0) node (a3) {} edge (a2)
       (2,0) node[fill=black] (a4) {} edge (a3)
       (3,0) node (a5) {} edge (a4)
       (3,1) node (a6) {} edge (a5)
       (4,0) node (a7) {} edge (a5)
       (5,0) node (a8) {} edge (a7)
       (6,0) node (a9) {} edge (a8)
       (6,1) node (a10) {} edge (a9)
       (7,0) node (a11) {} edge (a9)
       (2,1) node (a12) {} edge (a4);
       \node[draw=none,fill=none,below] at ($(a4)-(0,0.1)$) {$x$};
       \node[draw=none,fill=none] at (-1,0) {$T$};
  \end{scope}
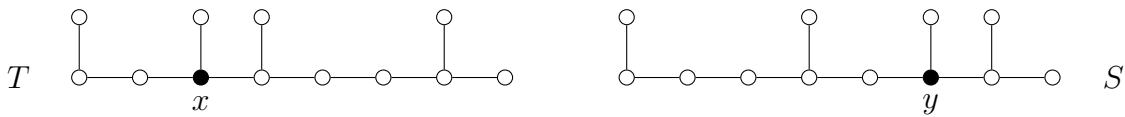
  \begin{scope}[scale=.8, xshift=.6\textwidth]
\path (0,1) node (a1) {}
       (0,0) node (a2) {} edge (a1)
       (1,0) node (a3) {} edge (a2)
       (2,0) node (a4) {} edge (a3)
       (3,0) node (a5) {} edge (a4)
       (3,1) node (a6) {} edge (a5)
       (4,0) node (a7) {} edge (a5)
       (5,0) node[fill=black] (a8) {} edge (a7)
       (6,0) node (a9) {} edge (a8)
       (6,1) node (a10) {} edge (a9)
       (7,0) node (a11) {} edge (a9)
       (5,1) node (a12) {} edge (a8);
       \node[draw=none,fill=none,below] at ($(a8)-(0,0.1)$) {$y$};
       \node[draw=none,fill=none] at (8,0) {$S$};
  \end{scope}
\end{tikzpicture}

\medskip

\noindent
Specifically, $T=L_x\cdot M_z$ and $S=L_y\cdot M_z$ where $M=P_2$ is the single-edge graph and
$L$ is the Harary-Palmer tree with non-similar strongly walk-equivalent vertices $x$ and $y$;
see Figure \ref{fig:HP}(a). Thus, $x$ is an ambivalent vertex in $T$ and $y$ is an ambivalent
vertex in $S$. The trees $T$ and $S$ have 12 vertices.
The computation shows that in every other tree with at most 12 vertices, all vertices are decisive.
\end{example}

\begin{example}[The smallest ``sporadic'' example]
  Similarly to the preceding example, we can construct two pairs of 13-vertex non-isomorphic trees
  with strongly walk-equivalent vertices. The tree $M=P_3$ has now 3 vertices. The root of $M$ can be
  chosen in two different ways, which gives us two different pairs. The computation reveals
  one more pair of trees $T$ and $S$ with walk-equivalent (but not strongly walk-equivalent)
  vertices $x\in V(T)$ and $y\in V(S)$:

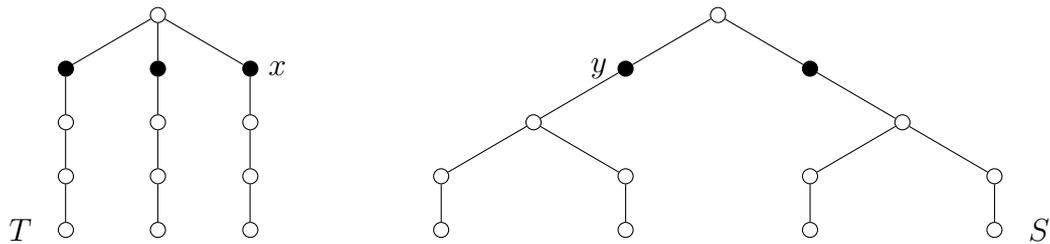
\bigskip

\begin{center}
\begin{tikzpicture}

  \matrix[column sep=1cm,row sep=.5cm,every node/.style={circle,draw,inner sep=2pt,fill=none,draw}] {
    &\node (a) {};&\\
\node[fill=black] (b1) {};&\node[fill=black] (b2) {};&\node[fill=black] (b3) {};\\
\node (c1) {};&\node (c2) {};&\node (c3) {};\\
\node (d1) {};&\node (d2) {};&\node (d3) {};\\
\node (e1) {};&\node (e2) {};&\node (e3) {};\\
    };

\node[draw=none,fill=none,right] at ($(b3)+(0.1,0)$) {$x$};

\draw  (a) -- (b1) -- (c1) -- (d1) -- (e1);
\draw  (a) -- (b2) -- (c2) -- (d2) -- (e2);
\draw  (a) -- (b3) -- (c3) -- (d3) -- (e3);

\node[draw=none,fill=none,left] at ($(e1)-(0.3,0)$) {$T$};

\end{tikzpicture}\qquad\qquad
\begin{tikzpicture}

  \matrix[column sep=1cm,row sep=.5cm,every node/.style={circle,draw,inner sep=2pt,fill=none,draw}] {
    &&&\node (a) {};&&&\\
&&\node[fill=black] (b1) {};&&\node[fill=black] (b2) {};&&\\
&\node (c1) {};&&&&\node (c2) {};&\\
\node (d1) {};&&\node (d2) {};&&\node (d3) {};&&\node (d4) {};\\
\node (e1) {};&&\node (e2) {};&&\node (e3) {};&&\node (e4) {};\\
    };

\node[draw=none,fill=none,left] at ($(b1)-(0.1,0)$) {$y$};

\draw  (a) -- (b1) -- (c1) -- (d1) -- (e1);
\draw  (a) -- (b2) -- (c2) -- (d4) -- (e4);
\draw  (c1) -- (d2) -- (e2);
\draw  (c2) -- (d3) -- (e3);

\node[draw=none,fill=none,right] at ($(e4)+(0.3,0)$) {$S$};

\end{tikzpicture}
\end{center}

\medskip

\noindent
This example is particularly interesting because it is unrelated to any
construction described in Section \ref{s:trees}.
Its distinguishing feature is that the walk-equivalent vertices $x$ and $y$
are not closed-walk-equivalent.

Note that \cite[Example (ii)]{KnopMSRT83} shows a ``sporadic'' example of non-isomorphic
(even not cospectral) trees on 12 vertices containing closed-walk-equivalent vertices.
\end{example}

\begin{example}\label{ex:Schwenk-and-anti} (The smallest examples showing that closed-walk-equivalent vertices in a graph
  do not need to be walk-equivalent and vice versa).
  The vertices $x$ and $y$ in the Schwenk graph~\cite{Schwenk73}

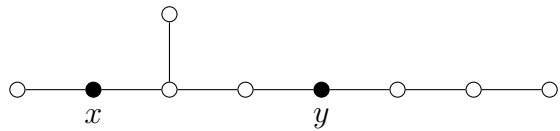
\bigskip

\begin{center}
\begin{tikzpicture}[every node/.style={circle,draw,inner sep=2pt,fill=none}]
\path (0,0) node (a1) {}
       (1,0) node[fill=black] (a2) {} edge (a1)
       (2,0) node (a3) {} edge (a2)
       (3,0) node (a4) {} edge (a3)
       (4,0) node[fill=black] (a5) {} edge (a4)
       (5,0) node (a6) {} edge (a5)
       (6,0) node (a7) {} edge (a6)
       (7,0) node (a8) {} edge (a7)
       (2,1) node (a9) {} edge (a3);
       \node[draw=none,fill=none,below] at ($(a2)-(0,0.1)$) {$x$};
       \node[draw=none,fill=none,below] at ($(a5)-(0,0.1)$) {$y$};
\end{tikzpicture}
\end{center}

\medskip

\noindent
are closed-walk-equivalent but not walk-equivalent. On the other hand,
the vertices $x$ and $y$ in the graph

\bigskip

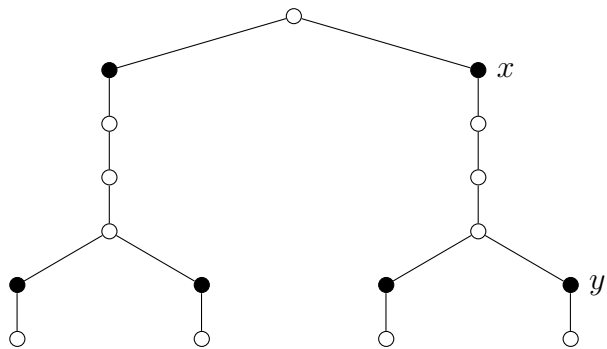
\begin{center}
\begin{tikzpicture}

  \matrix[column sep=1cm,row sep=.5cm,every node/.style={circle,draw,inner sep=2pt,fill=none,draw}] {
    &&&\node (a) {};&&&\\
    &\node[fill=black] (b1) {};&&&&\node[fill=black] (b2) {};&\\
    &\node (c1) {};&&&&\node (c2) {};&\\
    &\node (d1) {};&&&&\node (d2) {};&\\
    &\node (e1) {};&&&&\node (e2) {};&\\
\node[fill=black] (f1) {};&&\node[fill=black] (f2) {};&&\node[fill=black] (f3) {};&&\node[fill=black] (f4) {};\\
\node (g1) {};&&\node (g2) {};&&\node (g3) {};&&\node (g4) {};\\
    };

\node[draw=none,fill=none,right] at ($(b2)+(0.1,0)$) {$x$};
\node[draw=none,fill=none,right] at ($(f4)+(0.1,0)$) {$y$};

\draw  (a) -- (b1) -- (c1) -- (d1) -- (e1) -- (f1) -- (g1);
\draw  (a) -- (b2) -- (c2) -- (d2) -- (e2) -- (f4) -- (g4);
\draw  (e1) -- (f2) -- (g2);
\draw  (e2) -- (f3) -- (g3);

\end{tikzpicture}
\end{center}
\end{example}

\medskip

\noindent
are walk-equivalent but not closed-walk-equivalent.

\begin{example}[Trees of different size containing strongly walk-equivalent vertices $x$ and $y$]\label{ex:dist-size-strong}
We also note that the pair of vertex invariants $(W_G(v),R_G(v))$, in general, does not allow us
to determine the number of vertices of $G$, even in the case of trees.
This is demonstrated by two trees $T$ and $S$ with 11 and 10 vertices respectively
containing vertices $x\in V(T)$ and $y\in V(S)$ such that $W_T(x)=W_S(y)$ and $R_T(x)=R_S(y)$:

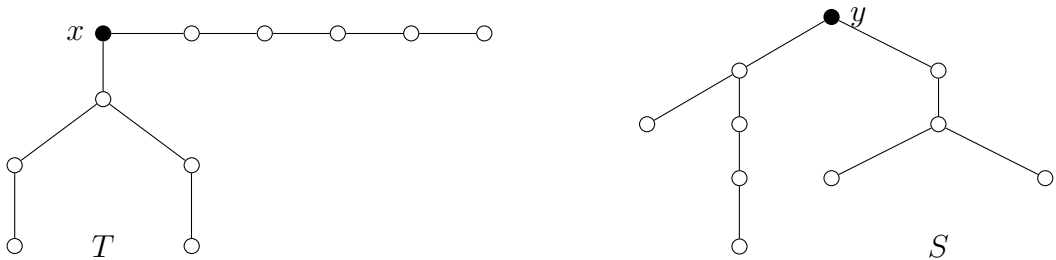
\bigskip

\begin{tikzpicture}

  \matrix[column sep=.75cm,row sep=.66cm,every node/.style={circle,draw,inner sep=2pt,fill=none,draw}] {
&\node[fill=black] (a1) {};&\node (a2) {};&\node (a3) {};&\node (a4) {};&\node (a5) {};&\node (a6) {};\\
&\node (b) {};&&&&&\\
\node (c1) {};&&\node (c2) {};&&&&\\
\node (d1) {};&\node[draw=none] (T) {$T$};&\node (d2) {};&&&&\\
    };

\node[draw=none,fill=none,left] at ($(a1)-(0.1,0)$) {$x$};

\draw  (d1) -- (c1) -- (b) -- (c2) -- (d2)  (b) -- (a1) -- (a2) -- (a3) -- (a4) -- (a5) -- (a6);

\end{tikzpicture}\qquad\qquad
\begin{tikzpicture}

    \matrix[column sep=1cm,row sep=.5cm,every node/.style={circle,draw,inner sep=2pt,fill=none,draw}] {
&&\node[fill=black] (a) {};&&\\
&\node (b1) {};&&\node (b2) {};&\\
\node (c0) {};&\node (c1) {};&&\node (c2) {};&\\
&\node (d1) {};&\node (d21) {};&&\node (d22) {};\\
&\node (e1) {};&&\node[draw=none] (S) {$S$};&\\
    };

\node[draw=none,fill=none,right] at ($(a)+(0.1,0)$) {$y$};

\draw (e1) -- (d1) -- (c1) -- (b1) -- (a) -- (b2) -- (c2) -- (d22) (b1) -- (c0) (c2) -- (d21);

\end{tikzpicture}

\bigskip

This is a single pair of this kind among trees with at most 11 vertices.
\end{example}

\begin{example}
There are smaller trees of different size containing closed-walk-equivalent vertices $x$ and $y$.
The trees $P_7$ and $Y_5$ are a unique pair of this kind among trees with at most 7 vertices:

\bigskip

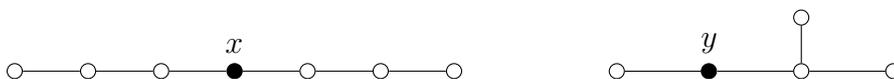
\begin{tikzpicture}
  \matrix[column sep=.75cm,row sep=.5cm,every node/.style={circle,draw,inner sep=2pt,fill=none,draw}] {
\node (a1) {};&\node (a2) {};&\node (a3) {};&\node[fill=black] (a4) {};&\node (a5) {};&\node (a6) {};&\node (a7) {};\\
    };
\node[draw=none,fill=none,above] at ($(a4)+(0,0.1)$) {$x$};
\draw  (a1) -- (a2) -- (a3) -- (a4) -- (a5) -- (a6) -- (a7);
\end{tikzpicture}\qquad\qquad
\begin{tikzpicture}

  \matrix[column sep=1cm,row sep=.5cm,every node/.style={circle,draw,inner sep=2pt,fill=none,draw}] {
&&\node (b) {};&\\
\node (a1) {};&\node[fill=black] (a2) {};&\node (a3) {};&\node (a4) {};\\
    };

\node[draw=none,fill=none,above] at ($(a2)+(0,0.1)$) {$y$};

\draw (a1) -- (a2) -- (a3) -- (a4) (a3) -- (b);

\end{tikzpicture}

\bigskip

This example is also interesting in the following respect. As easily seen,
vertices $v\in V(G)$ and $u\in V(H)$ in two graphs $G$ and $H$ are (closed-)walk-equivalent
  if and only if they are (closed-)walk-equivalent also in the (single)
graph obtained from $G$ and $H$ by adding a new edge between $u$ and~$v$.
In particular, if we connect the vertices $x$ and $y$ in the picture above,
then we obtain a tree on 12 vertices with two non-similar closed-walk-equivalent vertices.
This is the tree $E_6$ shown in \cite{RandicK89}
along with other three trees on 12 vertices having this property.

\end{example}

\begin{example}
  \label{ex:dist-size}
Finally, we show two trees with 8 and 11 vertices containing walk-equivalent vertices $x$ and~$y$.

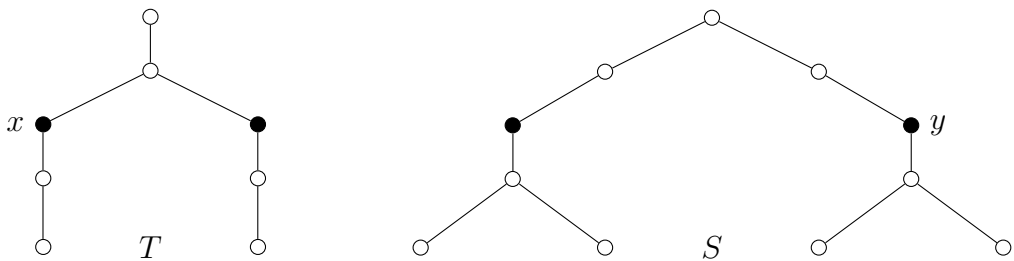
\bigskip

\begin{tikzpicture}

  \matrix[column sep=1cm,row sep=.5cm,every node/.style={circle,draw,inner sep=2pt,fill=none,draw}] {
&\node (c) {};&\\
&\node (a) {};&\\
\node[fill=black] (b1) {};&&\node[fill=black] (b2) {};\\
\node (c1) {};&&\node (c2) {};\\
\node (d1) {};&\node[draw=none] (T) {$T$};&\node (d2) {};\\
    };

\node[draw=none,fill=none,left] at ($(b1)-(0.1,0)$) {$x$};

\draw  (d1) -- (c1) -- (b1) -- (a) -- (b2) -- (c2) -- (d2) (a) -- (c);

\end{tikzpicture}\qquad\qquad
\begin{tikzpicture}

    \matrix[column sep=1cm,row sep=.5cm,every node/.style={circle,draw,inner sep=2pt,fill=none,draw}] {
&&&\node (a) {};&&&\\
&&\node (b1) {};&&\node (b2) {};&&\\
&\node[fill=black] (c1) {};&&&&\node[fill=black] (c2) {};&\\
&\node (d1) {};&&&&\node (d2) {};&\\
\node (e11) {};&&\node (e12) {}; &\node[draw=none] (S) {$S$};&  \node (e21) {};&&\node (e22) {};\\
    };

\node[draw=none,fill=none,right] at ($(c2)+(0.1,0)$) {$y$};

\draw (e11) -- (d1) -- (c1) -- (b1) -- (a) -- (b2) -- (c2) -- (d2) -- (e22) (d1) -- (e12) (d2) -- (e21);

\end{tikzpicture}

\bigskip

This is the smallest pair of this kind among trees with at most 11 vertices.
There are three more such pairs where the larger tree has 11 vertices.
One of these pairs is presented in Example \ref{ex:dist-size-strong},
and the smaller tree in the other two such pairs has 9 vertices.

\medskip

In conclusion, we make an algebraic-linear analysis of the last example as an illustration of
the framework presented in Section \ref{ss:wm}. In the notation of Section \ref{ss:wm}, let us fix $S=V(G)$.
The number $r$ and the sequence $a_1,\ldots,a_r$ defined by Eq.~\refeq{AAA} have a spectral interpretation,
which we briefly explain for the completeness of exposition.
An eigenvalue $\mu$ of $A$ is called \emph{main} if the eigenspace of $\mu$
is not orthogonal to $\jj$, where $\jj$ denotes the all-ones vector. The \emph{main polynomial} of $G$ is defined by
$$
M_G(z)=\prod_i(z-\mu_i),
$$
where the product is taken over all distinct main eigenvalues of $G$.
Lemma 3.9.8 in \cite{CvetkovicRS10} says that for any polynomial $p\in\bR[z]$,
the equality $p(A)\,\jj = 0$ is true if and only if $p$ is divisible by $M_G$.
This immediately implies that the main polynomial coincides with the polynomial
defined by Eq.~\refeq{main-poly}.
It is known \cite{Rowlinson07} that $M_G\in\bZ[z]$, i.e.,
all coefficients $a_1,\ldots,a_r$ are integers.

Now, looking at the smallest linear relation
between the columns of the
walk matrices of $T$ and $S$, we see that
$$
M_T(z)=z^4-z^3-4z^2+4z=(z+2)(z-2)(z-1)z
$$
and
$$
M_S(z)=z^5-6z^3+8z=(z^2-2)(z+2)(z-2)z.
$$
Thus, the main eigenvalues are $2,1,0,-2$ for $T$ and $2,\sqrt2,0,-\sqrt2,-2$ for~$S$.

Let  $\chi_x$ and $\chi_y$ denote the characteristic polynomials of the sequences
$W_T(x)$ and $W_S(y)$ respectively. Of course, $\chi_x=\chi_y$ as $W_T(x)=W_S(y)$.
As argued in the proof of Lemma \ref{lem:chiP},
$\chi_x=\chi_y$ is a common divisor of $M_T$ and $X_S$, i.e., of the polynomial $(z+2)(z-2)z$.
The sequence of walk numbers for $x$ (and $y$) is
\begin{equation}
  \label{eq:walk-seq}
1,\    2,\     5,\     8,\    20,\    32,\    80,\   128,\   320,\   512,\  1280,\  2048,\ \ldots
\end{equation}
We use the general fact that the order of a linear recurrent sequence is equal to the rank of
its Hankel matrix \cite{Gantmacher}. Since the Hankel matrix of the sequence \refeq{walk-seq}
has rank 3, we conclude that $\chi_x(z)=\chi_y(z)=(z+2)(z-2)z=z^3-4z$.
This is the characteristic polynomial of the linear recurrence $w_k=4w_{k-2}$. Consequently,
the walk number sequence \refeq{walk-seq} splits, after removal of the first element,
into two geometric progressions (one in the odd positions and the other in the even positions).
\end{example}

\end{document}